\documentclass[11pt,letterpaper]{article}
\usepackage[margin=1in]{geometry}

\usepackage{amsthm}
\usepackage{multicol}

\usepackage{colortbl}
\usepackage{xspace}
\usepackage{mdframed}
\usepackage{color,soul}

\usepackage{paralist}
\usepackage{pdfpages}
\usepackage{enumitem}
\usepackage{float}
\usepackage{booktabs}
\usepackage[override]{cmtt}
\usepackage{color,xcolor}
\usepackage{authblk}
\usepackage{graphicx,color,eso-pic}
\usepackage{pgfplots, pgfplotstable}

\sethlcolor{lightgray}

\usepackage{cite}
\usepackage{amsthm,amstext}
\usepackage{amsmath,amstext,amsfonts,amssymb,latexsym}
\usepackage{amsbsy}
\usepackage{stmaryrd}
\usepackage{pifont}

\definecolor{darkred}{rgb}{0.5, 0, 0}

\usepackage[bookmarks=true,pdfstartview=FitH,colorlinks,linkcolor=darkred,filecolor=darkred,citecolor=darkred,urlcolor=darkred]{hyperref}

\usepackage[capitalize,noabbrev]{cleveref} 

\ifdefined\ACM
\else
\usepackage{amsthm}
\fi

\usepackage{colortbl}
\usepackage{mdframed}
\usepackage{color,soul}
\usepackage{paralist}
\usepackage{pdfpages}
\usepackage{enumitem}
\usepackage{float}
\usepackage{booktabs}
\usepackage[override]{cmtt}
\usepackage{wrapfig}
\usepackage{boxedminipage}
\usepackage{url}
\usepackage{graphicx}
\usepackage{xspace}
\usepackage{multirow}
\usepackage{amsmath,amstext,amsfonts,latexsym}
\usepackage{amsbsy}
\usepackage{stmaryrd}
\usepackage{pifont}
\usepackage{color,xcolor}
\usepackage{graphicx,color,eso-pic}
\usepackage[capitalize,noabbrev]{cleveref} 

\sethlcolor{lightgray}

\definecolor{darkred}{rgb}{0.5, 0, 0}
\definecolor{darkgreen}{rgb}{0, 0.5, 0}
\definecolor{darkblue}{rgb}{0,0,0.5}

\newcommand\markx[2]{}

\newcommand{\ignore}[1]{}

\newcommand{\myparagraph}[1]{\vspace{10pt}\noindent\textbf{#1}}

\definecolor{darkgreen}{rgb}{0,0.5,0}
\definecolor{lightblue}{RGB}{0,176,240}
\definecolor{darkblue}{RGB}{0,112,192}
\definecolor{lightpurple}{RGB}{124, 66, 168}
\definecolor{grey}{RGB}{139, 137, 137}
\definecolor{maroon}{RGB}{178, 34, 34}
\definecolor{green}{RGB}{34, 139, 34}
\definecolor{types}{RGB}{72, 61, 139}
\definecolor{gold}{rgb}{0.8, 0.33, 0.0}

\definecolor{darkgray}{gray}{0.3}

\newcounter{task}

\ifdefined\ACM
\else
    \newtheorem{thm}{Theorem}[section]      
    \newtheorem{theorem}[thm]{Theorem}
    \newtheorem{conj}[thm]{Conjecture}
    \newtheorem{conjecture}[thm]{Conjecture}
    \newtheorem{lemma}[thm]{Lemma}
    \newtheorem*{lemma*}{Lemma}
    \newtheorem{claim}[thm]{Claim}
    \newtheorem{corollary}[thm]{Corollary}
    \newtheorem{fact}[thm]{Fact}
    \newtheorem{proposition}[thm]{Proposition}
    \newtheorem{example}[thm]{Example}

    \newtheoremstyle{boxes}
    {2pt}
    {0pt}
    {}
    {}
    {\bfseries}
    {}
    {\newline}
    {\thmname{#1}\thmnumber{ #2}:  
    \thmnote{#3}}

    \theoremstyle{boxes}

    {
    \theoremstyle{definition}
    \newtheorem{definition}{Definition}
    \newtheorem{remark}{Remark}
    }
\fi

\crefname{fact}{fact}{facts}
\Crefname{fact}{Fact}{Facts}

\ifdefined\authorNote
\newcommand{\elaine}[1]{{\footnotesize\color{magenta}[Elaine: #1]}}

\newcommand{\ke}[1]{{\footnotesize\color{red}[Ke: #1]}}
\else
\newcommand{\elaine}[1]{}
\newcommand{\hao}[1]{}
\newcommand{\ke}[1]{}
\fi

\ignore{
\newtheorem{thm}{Theorem}[section]      

\newtheorem{theorem}[thm]{Theorem}

\newtheorem{lemma}[thm]{Lemma}
\newtheorem{claim}[thm]{Claim}

\newtheorem{fact}[thm]{Fact}

\newtheorem{remark}[thm]{Remark}

\theoremstyle{definition}
\newtheorem{definition}[thm]{Definition}
}

\newcommand{\amt}{\ensuremath{v}\xspace}
\newcommand{\id}{\ensuremath{u}\xspace}
\newcommand{\aux}{\ensuremath{\alpha}\xspace}
\newcommand{\potential}{\ensuremath{\Phi}}

\begin{document}
\begin{titlepage}
\title{Mechanism Design for Automated Market Makers}
\author{T-H. Hubert Chan \qquad Ke Wu \qquad Elaine Shi\thanks{Author ordering is randomized.}}
\date{}

\maketitle
\thispagestyle{empty}

\begin{abstract}
Blockchains have popularized automated market makers (AMMs), applications that run on a blockchain, maintain a pool of crypto-assets, and execute trades with users governed by some pricing function. AMMs have also introduced a significant challenge known as the Miner Extractable Value (MEV). Specifically, 
miners who control the contents and sequencing of 
transactions in a block can extract value by front-running and back-running users' transactions, creating arbitrage opportunities that 
guarantee them risk-free returns.
MEV not only harms ordinary users, but more critically, 
encourages miners to auction off favorable transaction placements 
to users and arbitragers. This has fostered a more centralized 
off-chain eco-system, departing from the decentralized equilibrium 
originally envisioned for the blockchain infrastructure layer. 

In this paper, we consider how to design AMM mechanisms 
that eliminate MEV opportunities. 
Specifically, we propose a new AMM mechanism that processes all transactions contained within a block according to some pre-defined rules, ensuring
that some constant potential function is maintained
after processing the batch. 
We show that our new mechanism satisfies two tiers of guarantees. First, for legacy blockchains where each block is proposed by a single (possibly rotating) miner, we prove that our mechanism satisfies arbitrage resilience, i.e., a miner cannot gain risk-free profit. 
Second, for blockchains where the block proposal process is decentralized and offers sequencing-fairness, we prove a strictly stronger notion called incentive compatibility --- roughly speaking, we guarantee that any individual user's best response is to follow the honest strategy.

Our results complement prior works on MEV resilience in the following senses.
First, prior works 
have shown impossibilities to address MEV entirely at the consensus level.
Our work demonstrates a new paradigm of mechanism design
at the application (i.e., smart contract) layer to ensure 
provable guarantees of incentive compatibility.
Second, many works 
have attempted to augment the underlying consensus
protocol with extra properties such as sequencing fairness.
While most previous works 
heuristically argued why 
these extra properties help to mitigate MEV, 
our work demonstrates in a mathematically formal 
manner how to 
leverage such consensus-level properties 
to aid the design of incentive-compatible mechanisms.

\end{abstract}

\end{titlepage}

\tableofcontents
\thispagestyle{empty}
\newpage
\setcounter{page}{1}

\section{Introduction}

Blockchains have popularized 
decentralized finance (DeFi), with one of its key applications
being  
Decentralized Exchanges (DEX) based on 
Automatic Market Makers (AMMs)~\cite{theoryamm}. 
As of March 2021, the top six 
AMMs, including Uniswap, Balancer, and others, 
collectively held approximately  
\$15 billion in crypto assets~\cite{sok-amm}. A typical AMM exchange 
maintains a pool of capital called the ``liquidity pool''
with two crypto-assets $X$ and $Y$. 
A smart contract specifies the rules how users can trade assets
with the pool.
For example, one commonly adopted rule is a 
{\it constant-product
potential function} defined as follows. 
Let ${\sf Pool}(x, y)$ 
denote the pool's state where $x \geq 0$ and $y \geq 0$
represent the units of $X$ and $Y$  
held by the pool, respectively.
A constant product potential requires that $x \cdot y = C$
for some constant $C > 0$.
This means that if a user buys $\delta x$ amount of $X$
from the pool, it needs to pay $-\delta y$  
amount of $Y$ such that $(x - \delta x) (y - \delta y) = C$.

DeFi applications such as AMMs have introduced
opportunities for miners 
to profit, often in a risk-free manner,
 by front-running and/or back-running 
the users' transactions, a phenomenon 
known as Miner Extractable Value (MEV).
Despite the decentralized nature of blockchains, 
the block proposal process in mainstream consensus protocols 
remains centralized.
For each block, a single selected miner\footnote{In this paper, 
no matter whether the underlying consensus 
is proof-of-stake or proof-of-work, we
generically refer to a consensus node
that produces blocks as a ``miner'' or a ``block producer''.} has unilateral control 
over 
which transactions are included and their sequencing. 
By exploiting this capability, 
miners can profit, often in a risk-free manner.  
For example, in a 
sandwich attack~\cite{darkforest,sok-amm,theorymev,highfreq-dex},
a miner identifies a victim user attempting to purchase
a crypto asset $X$ at a maximum price of $r$, 
and inserts
a ${\sf Buy}(X)$ transaction just before
the victim's buy order 
and a ${\sf Sell}(X)$ transaction
immediately after.
Since purchasing $X$ increases
its price, 
the miner effectively buys at a lower price
through front-running, forces
the victim to buy at the worst possible price $r$, 
and then sells at a higher price through back-running, 
locking in a profit. 
Beyond sandwich attacks,
miners can also take advantage of  
more sophisticated arbitrage 
opportunities to profit~\cite{just-in-time-defi,theorymev,darkforest}.

MEV is widely recognized as one of the most important 
challenges for blockchains  
today for several reasons. 
First, since MEV is extracted at the expense of users, it effectively increases
the barrier of entry for ordinary users to engage with 
DeFi applications.
Second, 
MEV undermines the 
stability and security of the underlying 
consensus protocol~\cite{instability,a2mm,highfreq-dex,darkforest}.
Specifically, miners may be incentivized 
to fork the blockchain if doing so offers
higher MEV rewards than standard block rewards. 
Third, the block producer's power
in deciding the block contents and sequencing 
has given rise to an 
off-chain economy. Block producers enter private
contracts with arbitragers and users alike, 
offering them favorable positions in the 
block. 
\ignore{1) the ability to front-run/back-run transactions to 
searchers who specialize in identifying arbitrage opportunities,
and 2) the ability to be protected from
frontrunning/backrunning to ordinary users. 
}
These private off-chain contracts have led to 
a centralizing effect in the 
underlying layer 1 (i.e., the consensus layer), 
causing the {\it de facto} layer 1~\cite{gupta2023centralizing}
to operate in a manner that significantly departs  
from the intended design, 
and its equilibrium behavior is not understood. 
A recent empirical 
measurement showed that today, more than 85\% 
of the Ethereum blocks 
are built by two block producers~\cite{centraleth}.

\subsection{Our Results and Contributions}

It would have been compelling 
if there existed a way to solve the MEV problem
entirely at the consensus layer, without having to modify
the existing applications. 
Unfortunately, the impossibility results 
in several 
previous works~\cite{bahrani2023transaction,credible-ex}
can be interpreted to mean that 
solving the MEV 
problem (in its most general form) entirely at the 
consensus layer, subject to today's architecture, 
is impossible. 
On the other hand, 
many works aimed to offer strengthened
guarantees at  
the consensus layer 
such as 
sequencing fairness~\cite{espresso-seq,decentral-seq,orderfair00,orderfair01,orderfair02}
or some form of privacy~\cite{encryptmempool00,encryptmempool01}.
While it is widely believed that these properties
help to mitigate MEV, there has been relatively little formal investigation 
on how we can take advantage of these 
extra consensus-level properties
in mechanism design.

Therefore, in this paper, we ask  
the following natural questions: 
\begin{itemize}
\item 
Instead of working  
entirely at the consensus level, 
can we rely on mechanism design at the application level
(i.e., smart contract level) 
to obtain provable guarantees 
of MEV resilience?
\item 
How do 
strengthened 
guarantees at the consensus layer aid
the 
design of incentive compatible mechanisms
at the application layer?
\end{itemize}

Specifically, we ask these questions in the context
of Automated Market Makers (AMMs) which represent one of the most
important 
DeFi applications. 
We now summarize our results and contributions. 

\myparagraph{A mechanism design approach towards mitigating MEV.}
Applying the philosophy of mechanism
design at the application layer, we want to design 
an AMM mechanism 
that removes MEV opportunities and
provides incentive compatibility by construction. 

We devise a new AMM mechanism 
(to be executed as a smart contract on chain)
with the following abstraction. 
In our mechanism, the pool holds two crypto-assets $X$ and $Y$. 
A user can trade with the pool by posting
a buy/sell {\it order} specifying how much of $X$ (or $Y$) they want to buy/sell, and their worst acceptable exchange rate. 
When a new block arrives, the mechanism takes the block of orders 
as input, and applies an allocation rule to 
all orders contained within the block. 
The allocation rule decides
which orders are partially or completely satisfied, and at what price. 
The mechanism maintains the following invariant: 
the pool's beginning state
denoted ${\sf Pool}(x_0, y_0)$ and end state 
${\sf Pool}(x_1, y_1)$
is guaranteed to satisfy 
some ``natural'' potential function $\potential$
(e.g., 
the constant-product function mentioned above). 

Our mechanism offers two tiers of guarantees
depending on 
whether the properties of the underlying consensus. 
Specifically, we consider 
two models: 
1) the {\it plain} model, capturing
today's mainstream consensus protocols where
for any particular block, 
the inclusion and sequencing of transactions 
are determined by a single, 
possibly strategic block producer; and 2)
the {\it weak fair-sequencing} model, intended
to capture a new generation of 
consensus protocols 
that offer sequencing fairness 
guarantees~\cite{espresso-seq,decentral-seq,orderfair00,orderfair01,orderfair02}, 
e.g., through a decentralized sequencer. 
Our mechanism achieves the following desirable, two-tier properties: 

\begin{enumerate} 
\item 
{\bf Arbitrage free} in the 
plain model. We guarantee that 
no arbitrager (e.g., user, block producer, or any intermediary) 
can gain risk-free profit, 
{\it even when the arbitrager (e.g., block producer) has unilateral control
over the block contents and transaction sequencing}. 
Here, risk-free profit happens when an arbitrager 
can gain in one asset without losing in another with  
probability $1$.

\ignore{Moreover, we can guarantee a suitable notion of {\it fair treatment},
i.e., the mechanism does not make use
of the sequencing of transactions within the same block.  
In other words, the mechanism does not create opportunities 
for miners to auction off favorable positions in the block 
to arbitragers and users. 
}
\ignore{
\begin{itemize}[leftmargin=5mm,itemsep=1pt]
\item 
{\it Arbitrage free}: no one cannot profit from arbitrage opportunities.
\item 
{\it Fairness}: 
\elaine{TODO: we need a stronger notion of fairness.}
\end{itemize}
}
\item 
{\bf Incentive compatibility} 
in the weak fair-sequencing model. 
In the weak fair-sequencing model, 
our mechanism not only achieves arbitrage resilience,  
but 
also guarantees incentive compatibility.
Specifically, incentive compatibility means that users are incentivized
to report their true demand and true belief of  
the relative value of the two crypto-assets, 
and no strategic behavior
allows a user to gain.
Later, 
we prove that 
{\it incentive compatibility is a strictly stronger notion
than arbitrage resilience} (Fact~\ref{fct:SPimpliesAR}). 
 \end{enumerate}

Our weak fair-sequencing model is meant to capture 
a decentralized sequencer that sequences the orders
based on their arrival times (importantly, not based 
on the orders' submission time). 
While 
this model places additional constraints
on the strategy space in comparison with the plain model, 
it {\bf does not prevent
front-running} and thus {\bf does not trivialize}
the mechanism design problem. 
Notably, in this model, 
a strategic user or miner can still wait for a victim to submit its 
order, and then immediately submit a dependent order. 
{The strategic order 
can even front-run the victim's order} if the strategic user's network
is faster, 
similar
to a rushing attack in the cryptography literature~\cite{Canetti2000,uc}.


\ignore{assumes that each user has an intrinsic 
arrival time denoted $\alpha$ that is determined by nature,    
and the orders created by that user
cannot arrive earlier
than $\alpha$. 
We stress that even the weak fair-sequencing model {\bf does NOT 
trivialize the problem 
of achieving incentive compatibility} --- 
even with weak fair-sequencing, 
strategic users can still post orders dependent on honest users' orders
and front-run them. 
This can happen if the strategic user has a faster network connection, 
allowing it to observe honest users' orders, 
post dependent orders, and ultimately win the  
race against the honest orders\footnote{Similar
to a rushing attack in the cryptography literature~\cite{Canetti2000,uc}.}.
}

Our work is also {\it among the first 
to formally articulate, from a mechanism designer's perspective, 
how extra properties at the consensus level  
lend to the design of strategy-proof 
mechanisms at the smart contract layer}. 

\myparagraph{Conceptual contributions.}
We put forth new modeling and definitions, which capture 
a mechanism design problem of a decentralized nature. 
In comparison with the classical
mechanism design literature, 
our model and strategy spaces capture the ``permissionless''
nature of blockchains. 
\elaine{TODO: define strategy space more explicitly}
Specifically, 
a strategic player may not only report its 
valuation/demand untruthfully, but 
also inject fake orders or 
post multiple orders.
Such strategies are possible because the 
mechanism does not have a-priori
knowledge of the number or the identities of the bidders. 
Compared to closely related works, 
our model circumvents the strong impossibilities shown by 
Ferreira and Parkes~\cite{credible-ex}, 
because we relax some unrealistic restrictions
they impose --- see \Cref{sec:compare}
for more details.
Therefore, we believe that our model is better suited  
for capturing real-world mechanism design
at the smart contract layer, particularly for AMMs. 
\ignore{
our model
is more realistic and does not insist on the overly stringent requirement
that the mechanism running on the blockchain must be first-come-first-serve.  
Instead, we allow the mechanism designer to specify the mechanism
to be run on the blockchain. 
Because of this, we
can circumvent the strong impossibility results 
shown by Ferreira and Parkes~\cite{credible-ex}, and we therefore believe
that ours is a more suitable model for studying mechanism design
for DeFi applications.
}
Our new model and definitions naturally 
give rise to many interesting open questions 
which we discuss in \Cref{sec:open}.


\subsection{Comparison with Related Work}
\label{sec:compare}

\myparagraph{Comparison with most closely related work.} 
Ferreira and Parkes~\cite{credible-ex}
showed that in an overly pessimistic model as explained below, 
achieving arbitrage resilience is impossible, let alone 
incentive compatibility. 
However, their impossibility result  
holds only in an overly stringent model that does not reflect
the real-life design space.
Specifically, when interpreted in our new framework, 
their impossibility holds only if
the AMM mechanism has to respect the following specific structure:
sort the incoming orders according to some rules (called ``verifiable
sequencing rules'' in their paper), 
and run a legacy AMM contract 
that processes the sorted orders sequentially, 
such that the constant potential function
must be maintained after executing each order. 
In comparison, in our model, the constant potential function only
needs to be maintained at the end of processing the entire batch. 
Because of their strong impossibility result,
Ferreira and Parkes~\cite{credible-ex}
showed how to achieve a weaker guarantee in their model, that is, 
if the miner made risk-free profit, then the user should
enjoy a price that is at least as good as if its order
were the only one in the block. 
\ignore{The prior work of 
Ferreira and Parkes~\cite{credible-ex} recently proposed an approach 
called verifiable sequencing rules. In particular, they propose
that the block proposer must be subject
to a set of verifiable rules when constructing the block. 
However, {\it their work does not eliminate arbitrage, and still
allows the miner to make risk-free profit.}
The only provide a weaker type of guarantee, that is,
if the miner profited from arbitrage, the user has some price guarantee. 
Similarly, their work {\it does not provide incentive compatibility 
(even when the miner is not strategic 
and trusted to behave honestly)}, 
and users may be incentivized to bid untruthfully.
Ferreira and Parkes~\cite{credible-ex}
describe an impossibility result: it is not possible
to prevent the miner from making risk-free profit
in their model. However, {\it their impossibility 
holds only in an overly pessimistic model 
that is not realistic}. 
 In particular, 
their impossibility 
applies 
only if 1) the mechanism must 
be first-come-first-serve (FCFS), i.e., it must 
process the transactions one by one in the order that they appear in
the block; and 2) the block can 
contain an arbitrary {\it set} of transactions.
These assumption are 
too strong, since in practice, the mechanism 
we run on the blockchain (e.g., a smart contract) need not be 
FCFS and the mechanism designer
gets to design this smart contract; 
moreover, the blocks can have a validity
rule that precludes certain combinations of transactions. 
}

Li et al.~\cite{greedyseq}
inherit the same model as 
Ferreira and Parkes~\cite{credible-ex}, and they 
study what is the profit-maximizing strategy for the miner 
and the implications for the users when the miner 
adopts the optimal strategy. 
Like Ferreira and Parkes~\cite{credible-ex}, they adopt an 
overly restrictive model which requires them to  
give up 
on achieving arbitrage resilience, let alone incentive compatibility.

\ignore{
In comparison with these works, 
we take a step back and rethink what is a better model. 
{\it Our paper proposes a new model 
that is not only more realistic than 
Ferreira and Parkes~\cite{credible-ex} and 
also allows us to  
circumvent the strong impossibilities they show. For these
reasons,  we believe
it is a better model to study mechanism design for DeFi applications}.  
}
{\it One of the contributions we make is exactly to recognize 
why
the existing models are too stringent and unrealistic,
and suggest a better model for the study  
of 
MEV-resilient mechanism design}.

\myparagraph{Batch clearing at uniform price does not 
guarantee incentive compatibility.}
A line of works explored the idea of batch clearing
at uniform price. We stress that batch clearing 
at uniform price~\cite{batchamm00,batchamm01,cow} 
does not automatically
guarantee incentive compatibility.  
For example, suppose there are many eligible orders and the mechanism 
can clear only a subset of them. 
If the mechanism selects the subset based
on the declared valuation,  
then a strategic user can lie about its valuation  
to get selected.

The {concurrent and independent} work of Canidio and 
Fritsch~\cite{batchamm00,batchamm01}
suggested {batch-clearing at a uniform price}
and using a different potential function than 
Uniswap's constant-product function.
Their approach satisfies arbitrage resilience,
but it does not satisfy 
incentive compatibility {\it even when the miner 
is trusted to behave honestly}, for the reason stated above. 
In fact, Canidio and
Fritsch~\cite{batchamm00,batchamm01}
does not even fully specify which subset of orders to clear
when there are many eligible candidates. 
The work of Zhang et al.~\cite{batchnotic}
also observed that batch clearing alone does not imply incentive compatibility. In fact,
they investigated the optimal strategy under batch clearing. 
The prior work of Ramseyer et al.~\cite{ramseyer2023augmenting} 
also considered batch exchanges that clear at a uniform
price. 
Like Canidio and
Fritsch~\cite{batchamm00,batchamm01}, their work does not provide 
incentive compatibility, even when the miner is fully trusted.
Besides batch trading at uniform pricing, 
other forms of batch trading~\cite{fairbatchauction} have also been considered, 
but they also do not satisfy our notion of incentive compatibility. 

\myparagraph{Related works that do not address MEV.} 
Milionis, Moallemi, and Roughgarden~\cite{myersonian} consider how to design
the demand curve for a market maker to 
maximize profit and meanwhile incentivize truthful reporting. 
Their work is of a completely different nature than ours, since 
they {\it do not aim to address the problem of MEV}.
Specifically, they consider a simple model where 
users directly submit orders to the market maker.
They {\it do not consider any arbitrage strategy} where users
or miners try to front-run or back-run others' orders to make profit.

Bartoletti et al.~\cite{maximizemev}
studied the miner's optimal MEV strategy under transaction reordering. Their
work also does not provide a solution  
to mitigate or address MEV. 


\myparagraph{Understanding the impact of MEV.}
A line of works have 
empirically or theoretically investigated the 
profitability or impact of MEV~\cite{clockwork,theorymev,maximizingmev,darkforest,just-in-time-defi,analyzesandwich,attackdefi,bundleprofit}.

\myparagraph{Empirical approaches towards mitigating MEV.}
Another line of work
suggest that the users themselves take action to mitigate
MEV, either by setting their slippage limits more cleverly~\cite{a2mm}, 
or by exploiting arbitrage opportunities themselves
to lower their transactional costs~\cite{mevgametheory}.
There are also various blog posts  
on online forums that suggest alternative designs~\cite{surplusmaxamm00,surplusmaxamm01}.
However, these works are empirical and do not lend to the
theoretical understanding of 
the equilibrium behavior
of the eco-system.

Both academic research and real-world blockchain projects 
have made an effort to build decentralized 
sequencers~\cite{espresso-seq,decentral-seq,orderfair00,orderfair01,orderfair02}, 
or encrypted mempools~\cite{encryptmempool00,encryptmempool01}. 
The former approach removes the ability for a single block proposer 
to decide the block contents and sequencing, and 
achieves some form of sequencing fairness~\cite{orderfair00,orderfair01,orderfair02}\footnote{Sequencing fairness is also commonly referred to as ``order fairness''. In this paper,
we use the term ``sequencing fairness'' to avoid collision with the usage
of ``order'' to mean a trade proposal.}.
The latter  
approach allows users to submit transactions in 
committed or encrypted format, which makes it harder
for miners to front-run and back-run transactions.
However, 
from a mechanism design perspective, we still lack  
mathematical understanding 
to what extent these new consensus/cryptographic abstractions 
can help us mitigate MEV
and achieve strategy-proof DeFi mechanisms.
In this sense, our work is {\it among the first to mathematically
articulate 
how to rely on ``sequencing fairness'' 
to achieve strategy-proofness by construction}.

\myparagraph{Sequencing fairness.}
A line of works~\cite{espresso-seq,decentral-seq,orderfair00,orderfair01,orderfair02}
have studied how to achieve order fairness
in consensus. 
Numerous blockchain projects
are also building decentralized sequencers~\cite{decentralseq00,decentralseq01,decentralseq02,decentralseq03,decentralseq04,decentralseq05}
which 
can be one approach for 
achieving sequencing fairness. 
Some works~\cite{welfaregaporder00,welfaregaporder01} have also pointed out 
the price of sequencing fairness such as loss in welfare. 
Improving the underlying mechanisms for achieving sequencing fairness 
is outside the scope of this paper. 
We also leave it as future work to study how to optimize social welfare
under incentive compatibility.

\myparagraph{Other related works.}
There is a recent line of work on 
transaction fee mechanism 
(TFM)
design~\cite{roughgardeneip1559-ec,foundation-tfm,crypto-tfm}. 
This line of work 
aims to design mechanisms such that users, miners, and user-miner
coalitions are incentivized to behave honestly. 
However, the current modeling approach of this line of work
captures only the utilities at the consensus layer. 
They cannot capture ordering and application-level MEV. 
The recent work of Bahrani et al.~\cite{bahrani2023transaction}
showed strong impossibility results for fully solving this problem
at the TFM-layer alone.  
In this sense, our work complements 
the line of work on TFM design 
by taking an application-level  
(i.e., smart-contract-level)
approach 
towards achieving incentive compatibility by construction.

\subsection{Scope and Open Questions}
\label{sec:open}


Just like the recent literature including 
Ferreira and Parkes~\cite{credible-ex} and Li et al.~\cite{greedyseq}, 
the scope of the present paper is 
restricted to how defend against MEV 
in a {\it standalone} two-asset AMM mechanism. 
We begin with the standalone setting because
it serves as a necessary basis for understanding 
the compositional setting with multiple instances. 
It is also a widely adopted approach in the mechanism design 
and cryptography 
literature
to begin with the standalone setting first. 
We leave it as an interesting open question 
how to achieve provable game-theoretic
guarantees in a compositional setting   
where multiple instances can interact with each other. 
We stress that in general, unlike
in the cryptography literature where there 
are composable notions of security~\cite{Canetti2000,uc}
which makes composition worry-free, 
most notions in game theory do not naturally compose, 
and composition is typically treated on a case-by-case basis.
\ignore{
How to address cross-mechanism MEV,  
how to have game-theoretic 
properties that compose across multiple instances, 
and how to extend our 
results to more applications (e.g., multi-asset
AMMs, DeFi lending applications)
are left as open questions for future work.
}
\ignore{
Despite these limitations, 
we believe that our effort is a small but nonetheless
valuable step forward towards formulating a mathematical  
foundation 
for decentralized mechanism design, and particularly 
through the lens of a mechanism designer.
It is also worth mentioning that other 
subfields of decentralized mechanism design (e.g.,
the transaction fee mechanism 
literature~\cite{roughgardeneip1559-ec,foundation-tfm,crypto-tfm})
are also similarly in a nascent state. Likewise,
the initial foundational 
works~\cite{roughgardeneip1559-ec,foundation-tfm,crypto-tfm} in these areas  
also 
made similar simplifying assumptions
to make progress, 
e.g., they also start by exploring standalone setting and
leaving compositional guarantees 
and cross-block strategies to future work.
}


Our new model and definitions give rise to many interesting open problems.
One interesting question is how to extend our results to AMMs
with multiple assets.
Another interesting question is whether it is possible
to achieve the stronger notion of incentive compatibility 
without relying on sequencing fairness.
Currently, our model assumes that all the orders are submitted
in the clear, 
and we thus define incentive compatibility in the ex post setting.
A future direction is to understand how to define and achieve  
incentive compatibility
in an MPC-assisted model~\cite{crypto-tfm} 
or an ``encrypted mempool''~\cite{encryptmempool00,encryptmempool01} model 
where  
transactions are submitted in encrypted or committed format. 
Another interesting question is whether we can design strategy-proof
AMM mechanisms when the execution syntax is atomic rather than 
partial fulfillment like what we consider.

We believe that the modeling work in our paper
helps to lay the theoretical groundwork
for exploring questions such as above in the future. 







\section{Definitions}
\label{sec:defn}

\subsection{Swap Mechanism for AMMs}
A swap mechanism for a pair of assets $(X, Y)$ has a state 
(also called the {\it pool state})
denoted ${\sf Pool}(x, y)$ 
where $x$ and $y$ are non-negative values that represent the amount of  
each asset currently held by the mechanism.
A user can submit an order to trade with the mechanism
in two ways: either buy $X$ and pay in $Y$,
or buy $Y$ and pay in $X$.
Suppose the user buys $\delta x$ units of $X$
and pays $\delta y$ units of $Y$,
then the updated state 
after the trade 
will become ${\sf Pool}(x-\delta x, y + \delta y)$.

\myparagraph{Order.}
Each order is of the form 
$(t, \amt, r, \aux)$ where 
\begin{itemize} 
\item 
$t \in \{{\sf Buy}(X), {\sf Buy}(Y), {\sf Sell}(X), {\sf Sell}(Y)\}$ 
is the type of the order indicating
that the user wants to buy or sell 
and which asset; 
\item 
$\amt$
is a non-negative value that denotes 
the maximum amount of the user wants to buy or sell; 
\item 
$r$ denotes the user's 
acceptable exchange rate, i.e., 
the user believes that 
each unit of $X$ is worth $r$ units of $Y$.
For example, if the order is of type ${\sf Buy}(X)$, 
then the user is willing to pay at most $r$ 
units of $Y$ for each unit of $X$; 
if the type is ${\sf Sell}(Y)$, then $1/r$ is the minimum
asking price in $X$ for each unit of $Y$.
\ignore{
$r$ denotes the worst-case price 
the user is willing to accept for each asset. 
If it is a buy order, 
then the user is willing to buy the asset at any price
that is lower than or equal to the specified $r$.
Otherwise if it is a sell order,
the user is willing to sell the asset at any price
higher than or equal to the specified $r$.
}
\item 
$\aux$ is an arbitrary string denoting any additional auxiliary information,
e.g., the submitter's identity, timestamping information, 
position in the block, and so on. 
\end{itemize}

Note that given 
an order of the form (${\sf Sell}(Y)$, $\amt$, $r$, \_)\footnote{Here, the ignore symbol $\_$ means that we are ignoring the content of this field in the current context.},
\elaine{TODO: define this ignore symbol}
another way to view it is that the user 
wants to buy $X$; it 
is willing to pay at most $r$ units of $Y$ for each unit
of $X$; moreover, it wants to buy as many units of $X$ as possible
subject to a capital of $\amt$ units of $Y$.
Henceforth, for an order of the type ${\sf Buy}(X)$ or
${\sf Sell}(X)$, we say that $X$ is the 
{\it primary asset} of the order.

\ignore{
For example, we can express the following type of scenarios.
\begin{itemize}[leftmargin=5mm]
\item 
$({\sf Buy}(X), 5, 10, \infty)$:
a user wants to buy up to $5$ units of $X$ and willing
to pay at most $10$ units of $Y$ for each 
unit of $X$. The user has prepared at least $50$ units 
of $Y$ such that it can afford it even if the order
is fully executed; 
\item 
$({\sf Buy}(X), \infty, 10, 100)$: 
a user wants to invest in $X$ with a maximum budget of
$100$ units of $Y$, she is willing to pay at most $10$ units of $Y$
for each unit of $X$.
Another way to think about it is that the user wants
to sell up to $100$ units of $Y$, and 
she is willing to accept at least $0.1$ 
unit of $X$ for each unit of $Y$.
\end{itemize}
}

\elaine{TODO: check for v and change to $\amt$.}

\myparagraph{Swap mechanism.}
A possibly randomized (partial fulfillment) swap mechanism 
should define the following rules\footnote{In \Cref{sec:conclusion},
we pose the open question of how to achieve incentive compability for
all-or-nothing fulfillment mechanisms.}:
\begin{itemize} 
\item 
{\it Honest strategy.}
Given a user's private type $T$, the initial state ${\sf Pool}(x, y)$, 
the honest strategy, often denoted $HS(x, y, T)$, 
outputs a vector of orders the user should submit.
A user's private type $T$ can contain information
such as how many units of $X$ and $Y$ it currently holds,
and the user's private valuation of the exchange rate
between $X$ and $Y$.
\item 
{\it Allocation rule.} 
The allocation rule receives 
as input an initial state ${\sf Pool}(x, y)$,
a list of orders,  
and for each order $(t, \amt, r, \aux)$, it outputs
the following:
\elaine{TODO: add budget feasibility}
\begin{itemize}[leftmargin=5mm,itemsep=1pt]
\item 
the amount $\amt' \in [0, \amt]$
of primary asset that has been fulfilled --- note that the fulfillment
can be partial; 
\item 
an average exchange rate $r' > 0$ at which 
the order was fulfilled.
For a ${\sf Buy}(X)$ order,
it means that the user pays $\amt' \cdot r'$ units
of $Y$ in exchange for $\amt'$ units of $X$. 
For a ${\sf Sell}(Y)$ order, the user obtains $\amt'/r'$ units
of $X$ for the $\amt'$ units of $Y$ sold.
We require that for a ${\sf Buy}(X)/{\sf Sell}(Y)$
order, $r' \leq r$, i.e., the purchase price cannot
be higher than the specified maximum rate $r$; 
and for a 
${\sf Buy}(Y)/{\sf Sell}(X)$
order, $1/r' \leq 1/r$.

\end{itemize}
\end{itemize}

In a real-world instantiation, the pool state
is recorded on the blockchain, and the allocation rule
is executed on the blockchain (e.g., in the form of a smart contract).

\ignore{
We require that the mechanism satisfy 
individual rationality and budget feasibility:
\begin{enumerate}[leftmargin=6mm]
\item 
{\it Individual rationality.}
For a ${\sf Buy}(X)$ order, the average exchange
rate $r' \geq r$,  and for a ${\sf Buy}(Y)$ order, it must be $1/r' \geq 1/r$.
\item 
{\it Budget feasibility.}
Each order cannot pay more than the specified maximum budget $\beta$.
\end{enumerate}
}

\myparagraph{Invariant on pool state.}
We consider swap mechanisms
that satisfy the following invariant
on pool state.
Given some 
initial state ${\sf Pool}(x, y)$, and 
the outcome output by the allocation rule, one can uniquely
determine the ending state ${\sf Pool}(x', y')$. 
We require that 
initial and ending pool state must satisfy 
some constant potential function, that is, 
$\potential(x', y') = \potential(x, y)$. We define
potential function 
and requirements on the potential function below.

\myparagraph{Potential function.}
We consider swap mechanisms that  
respect a constant potential function $\potential(\cdot, \cdot)$. 
Specifically, suppose the pool's initial state 
is 
${\sf Pool}(x, y)$, 
and changes to 
${\sf Pool}(x', y')$ after 
the mechanism processes a batch of orders. 
Then, it must be that 
\[
\potential(x, y) = \potential (x', y').
\]
In practice, the most widely adopted approach is a constant-product
market maker where
$\potential(x, y) = x \cdot y$. 
In other words, 
suppose the initial pool state is $(x, y)$
and some user buys $\delta x$ amount of $X$,
then it must pay $-\delta y$ 
units of $Y$ where $\delta y$ 
can be calculated by solving the following equation:
\[
(x - \delta x) (y - \delta y) = x y.
\]

\elaine{TODO: remove reference to valid initial pool state}
\elaine{TODO: change all usage of potential func}

\myparagraph{Assumptions on the potential function.}
We assume the standard assumption that 
the potential function $\potential(\cdot, \cdot)$
is increasing, differentiable, and concave.

\myparagraph{Market exchange rate.}
In our swap mechanism, we will make use of the notion
of a market exchange rate, as defined below.
\begin{definition}[Market exchange rate]
\label{defn:exchange-rate}
Given a pool state ${\sf Pool}(x, y)$, the 
current market exchange rate 
is defined as 
\[
r(x, y) = \frac{\partial \potential/\partial x}{\partial \potential/\partial y}(x, y).
\]
Intuitively, it means that 
to buy an infinitesimally small $d x$
amount of $X$, we need to pay 
$r(x, y) \cdot dx $ units of $y$.
\end{definition}
Throughout this paper, whenever we say 
rate, it always means how much $y$ one has to pay per unit of $x$ rather than the other way around.

\ignore{
\begin{definition}[Valid potential function]
We assume that the potential function $\potential(x,y)$ satisfies the following conditions.
\begin{itemize}[leftmargin=5mm,itemsep=1pt]
\item 
For any possible $x\geq 0$ in the domain, there exists one and only one $y\geq 0$ such that $\potential(x,y) = 1$. 
Equivalently, $\potential(x,y)=1$ induces a function $y = \fy(x)$ such that $\potential(x, \fy(x)) = 1$.
\item 
$y = \fy(x)$ is a differentiable, decreasing and convex function:
\begin{itemize}[leftmargin=5mm,itemsep=1pt]
    \item Differentiable and decreasing: $\fy'(x)<0$.
    \item Convex: for any $x_0, x_1\in{\sf dom}(\fy)$, for any $\alpha\in[0,1]$, we have $\fy(\alpha x_0 + (1-\alpha) x_1)\leq \alpha \fy(x_0) + (1-\alpha) \fy(x_1)$, where ${\sf dom}(\fy)$ denotes the domain of $\fy$. \ke{do we need strictly convex}
\end{itemize}
\end{itemize}
\label{defn:potentialvalid}
\end{definition}
}

\subsection{Arbitrage Resilience}

Arbitrage resilience means that an arbitrager has no strategy such that it gets a net gain in one asset without any loss in the other.

\elaine{TODO: comment on why we don't allow Y to be 0 and X to remain.
1. it gives a stronger impossibility. 2.
with local efficiency, the two are equivalent}

\begin{definition}[Arbitrage resilience]
We say that a mechanism
satisfies 
arbitrage resilience iff 
given any initial pool state, 
any input vector of orders,
with probability $1$ over the random coins of the mechanism, 
the following must hold:
there does not exist a subset of orders whose joint outcomes  
result in 
$\delta x \geq 0$ net gain in $X$ and $\delta y \geq 0$
net gain in $Y$, such that 
at least one of $\delta x$ and $\delta y$ is strictly greater than $0$.
\ignore{
if any ${\sf Buy}(X)$ or ${\sf Sell}(Y)$ order is executed at an average exchange rate  
of $r$ and any 
${\sf Buy}(Y)$ or ${\sf Sell}(X)$ order is executed 
at an average 
exchange rate of $r'$, it must be that $r' \leq r$.
}
\label{defn:arbitrage}
\end{definition}

The definition above is consistent with 
Ferreira and Parkes~\cite{credible-ex}'s notion of no risk-free return,
although their scheme cannot guarantee no risk-free return (i.e.,
arbitrary resilience), whereas ours does. 
\elaine{TODO: double check}

\begin{remark}
If a mechanism satisfies \Cref{defn:arbitrage}, 
it means that it satisfies arbitrage resilience  
in a very strong sense: i.e.,   
an arbitrager (e.g., block producer)
cannot make risk-free profit 
{\it even when it  
can 1) fully control the block contents;
2) control the sequencing of orders within the block;
3) inject its own orders;
and 4) drop others' orders}.
\end{remark}

Indeed, the mechanism 
described later in this paper will satisfy 
arbitrage resilience 
even when the underlying consensus does not provide any sequencing fairness
guarantees.  


\ignore{
\begin{definition}[Fair Treatment]
We say that a mechanism satisfies fair treatment 
iff given an arbitrary vector of orders, 
for any two orders 
contained within this vector
whose $t$, $\amt$, and $r$ fields are identical, 
their outcomes are identically distributed.
\end{definition}
}

\ignore{
\subsection{Local Efficiency}

\begin{definition}[Local efficiency]
Given a vector of orders, we say that an outcome
is locally efficient iff 
there does not exist any buy order that is not completely fulfilled,
with a price $r$ is greater than the ending price $r^*$;  \elaine{TO DEFINE}
and there does not exist any sell order that is not completely fulfilled,
with a price $r$ smaller than the ending price $r^*$. 
We say that a mechanism is locally efficient, 
iff for any vector of orders, 
with probability $1$, the outcome is locally efficient. 
\end{definition}

\begin{definition}[Best possible local efficiency under arbitrage
resilience]
\elaine{TODO: FILL}
\end{definition}
}

\subsection{Plain Model}

The plain model is meant to capture
mainstream consensus protocol today where the contents of each individual
block is determined 
by a single block producer responsible for proposing or mining that block. 

\myparagraph{Strategy space in the plain model.}
In the plain model, 
we assume that a strategy user 
or miner with intrinsic type $(t, \amt, r, \aux)$
may engage in the following strategies:
\begin{itemize} 
\item 
Post zero or multiple arbitrary orders which may or may not
reflect its intrinsic type --- this captures 
strategies that involve misreporting valuation and demand, 
as well as posting
of fake orders; 
\item 
Censor honest users' 
orders --- this captures a strategic miner's ability to exclude 
certain orders from the 
assembled block;
\item  
Arbitrarily misrepresent
its own auxiliary information field $\alpha$, or  
even modify the $\alpha$ field of honest users' orders --- 
meant to capture the ability of a miner 
to decide the sequencing of the transactions within a block,
where the arrival-time and position information may be generically captured
by the auxiliary information field $\alpha$. 
\item 
Decide its strategy  
{\it after} having observed honest users' orders. 
\end{itemize}

\ignore{
Note that this general formulation of strategy space 
captures the miner's influence
too. For example, 
suppose the $\alpha$ in of the order 
encodes the position 
of the order 
within the block, 
and suppose that the 
mechanism makes use
such information make decisions, then a strategic
miner is able to 
arbitrarily modify $\alpha$ to its advantage. 
}
The coalition of a miner with 0 demand and  
a user with  some 
positive demand 
as captured by its type 
$(t, \amt > 0, r, \aux)$
can simply be viewed as a single strategic player with 
type $(t, \amt > 0, r, \aux)$.

\subsection{Weak Fair-Sequencing Model}
\label{sec:timedefn}


We define a weak fair-sequencing model,
meant to capture a new generation of consensus protocols
that employ a decentralized sequencer
and offers some form of sequencing fairness~\cite{espresso-seq,decentral-seq,orderfair00,orderfair01,orderfair02}. 
Such a decentralized sequencer
will sequence the transactions based
on their (approximate) arrival times. 
We stress that even in the weak fair-sequencing model,
it is possible for a strategic user to
observe a victim's order, post a dependent order,
and have the dependent order race against and front-run 
the victim's 
order. 
Such a front-running attack can succeed
especially 
when the strategic user's network is faster
than the victim's. 
In particular, we stress that the 
weak fair-sequencing model
is sequencing orders based on their {\it arrival times,
not the time of the submission of these orders}.

\ignore{
We introduce a weak fair-sequencing model 
that captures the scenario when the underlying consensus
layer has a decentralized block proposal process
that provides 
sequencing fairness~\cite{orderfair00,orderfair01,orderfair02}. 
Specifically, we assume that the  
proposed block promises to sequence the transactions according
to their arrival order.
We want to capture 
sequencing fairness without making an overly restrictive
assumption.  
In particular, we want  
to capture strategies where 
a strategic user or miner 
submits an order dependent on others' orders,
and 
meanwhile front-runs  
the orders that it depends on (e.g., if it has a faster
network link). 
We therefore use the following modeling approach. 
}
Recall 
that a user's intrinsic {\it type} is of the form $(t, \amt, r, \aux)$
where $(t, \amt, r)$
denotes the user's true valuation and 
budget. In the weak fair-sequencing model, we will use 
the $\aux$ field to encode 
the order's arrival time 
--- a smaller $\aux$ means that the user arrives earlier.

We shall assume that  under honest strategy, a user's order
should always be populated with the correct $\alpha$ whose
value is determined by nature,  
and equal to the time at which the order is generated 
plus the user's network
delay.  A strategic user is allowed 
to delay the submission of its order. 

\myparagraph{Strategy space in the weak fair-sequencing model.}
We consider the following strategy space
in the weak fair-sequencing model:
\begin{itemize} 
\item 
A strategic user or miner
with intrinsic type $(t, \amt, r, \aux)$
is allowed to post 
zero or multiple bids 
of the form $(\_, \_, \_, \aux')$
as long as $\aux' \geq \aux$.
This captures misreporting valuation and demand, 
posting fake orders, as well as delaying the posting
of ones' orders. 
\item 
The strategic user 
or miner 
can decide its strategy 
{\it after} observing honest users' orders. 
\end{itemize}

Compared to the plain model, the weak fair-sequencing model imposes some restrictions on the strategy space. 
Specifically, in the plain model, 
a strategic user or miner 
can arbitrarily modify the $\alpha$ field of its own order 
or even others' orders, 
and a strategic miner may censor honest users' orders. 
In the 
weak fair-sequencing model, a strategic 
user or miner can no longer
under-report its  
$\alpha$, 
cannot modify honest users' $\alpha$,  
and cannot censor honest users' orders, 
because
the sequencing of the transactions is determined 
by the underlying decentralized sequencer.

Importantly, 
despite these constraints on the strategy space,
the weak fair-sequencing model still permits
front-running-style attacks as mentioned earlier, and
thus mechanism design remains non-trivial even
under the slightly restricted strategy space.


\ignore{
\begin{remark}[Why the weak fair-sequencing model does NOT trivialize the problem]
Jumping ahead, our incentive compatibility notion
is in the {\it ex-post} setting, i.e., incentive compatibility 
is guaranteed even when 
a strategic player may post orders that depend on all other users' orders.
In this sense, our weak fair-sequencing model
permits {\it rushing} strategies: 
a strategic user 
can submit an order after seeing others' orders; 
and if its network is faster than honest users,
its order can actually arrive earlier than the honest ones!
In other words, the field $\aux$ should be viewed
as the time the order actually arrives, not when
the order was initially submitted to the network. 
\end{remark}
}

\subsection{Incentive Compatibility}

\myparagraph{Defining a user's preference among outcomes through a partial ordering.}
To define incentive compatibility, we first
need to define 
a ranking system that expresses
a user's preference among different outcomes.

We can use a pair $(\delta x, \delta y)$ to denote the outcome, 
meaning that the user has a net gain of $\delta x$ 
in $X$, and it has a net gain of $\delta y$ 
in $Y$ (where a net loss is captured as negative gain).
Consider two outcomes 
$(\delta x_0, \delta y_0)$ and $(\delta x_1, \delta y_1)$,
and suppose that the user's intrinsic type is 
$T = ({\sf Buy}(X), \amt, r, \_)$. 
Naturally, 
for such a user, 
outcome 
$(\delta x_1, \delta y_1)$ is at least as good as 
$(\delta x_0, \delta y_0)$, henceforth
denoted  
$(\delta x_0, \delta y_0) \preceq_T (\delta x_1, \delta y_1)$, if
one of the following is true: 
\begin{itemize} 
\item
$\delta x_0 \leq \delta x_1$, $\delta y_0 \leq \delta y_1$. 
In other words, 
relative to 
$(\delta x_0, \delta y_0)$, the user gains no less
in either asset  
in the latter outcome 
$(\delta x_1, \delta y_1)$.
\item 
$\delta x_0 \leq \delta x_1 \leq \amt$ (or $q \leq \delta x_1 \leq \delta x_0$), 
and $r(\delta x_1 - \delta x_0) \geq \delta y_0-\delta y_1$. 
In other words, 
the latter outcome $(\delta x_1, \delta y_1)$
is closer to satisfying the demand $\amt$, 
and moreover, the user paid at most $r$ {\it marginal} price
for each {\it extra} unit of $X$ 
in the latter outcome. 
\end{itemize}
\elaine{TODO: double check, this may not be complete}

\begin{remark}[Why define a partial ordering rather than 
a real-valued utility]
The reader may wonder why we do not define  
a real-valued utility which would have given a total ordering
on all outcomes. 
The reason why we define only a partial ordering
and allow some outcomes to be incomparable 
is because a strategic user 
(e.g., whose intrinsic demand is to buy up to $\amt$ units of $X$) 
can act arbitrarily, which may 
cause its net gain $\delta x$ in $X$ 
to be either negative, or greater than the intrinsic demand $\amt$.
We allow some of these outcomes to be incomparable.
For example,
suppose 
relative to $(\delta x_0, \delta y_0)$, 
the latter outcome $(\delta x_1, \delta y_1)$
buys some extra units at a margial price better than the specified
rate $r$, but it overshoots the intrinsic demand,
then these two outcomes are incomparable.

Finally, the case for 
other types including  
${\sf Buy}(Y)$, ${\sf Sell}(X)$, and ${\sf Sell}(Y)$ types, 
a partial ordering 
can be symmetrically defined --- we give the full
definition of the partial ordering in \Cref{sec:rank}.
Moreover, \Cref{sec:rank} 
also give more explanations why we choose to define a partial ordering
to rank the outcomes rather than a real-valued utility. 
\end{remark}

\elaine{TODO: refer to appendix for formal definition of partial ordering}

\myparagraph{Definition of incentive compatibility.}
Since in our paper, we consider deterministic mechanisms,
we will define incentive compatibility only for a deterministic
mechanism. Note that the definition can easily be extended
to randomized mechanism  
using suitable notions of stochastic dominance.

In the definition below,
we use $HS(T)$ to denote the honest strategy of a user
with intrinsic type $T$ --- for a direct-revelation mechanism, 
the honest strategy is simply to reveal the user's true type.
Further, we use ${\sf out}^u(x_0, y_0, {\bf b})$
to denote the outcome of user $u$ when
the mechanism is executed over initial pool
state ${\sf Pool}(x_0, y_0)$, and a vector of orders
${\bf b}$.

\begin{definition}[Incentive compatibility]
Given a deterministic swap mechanism, 
we say that it satisfies incentive compatibility (w.r.t. some partial
ordering relation $\preceq_T$), 
iff for any initial pool state ${\sf Pool}(x_0, y_0)$,
for any vector of orders ${\bf b}_{-u}$
belonging to all other users except $u$, 
for any intrinsic type $T$ of the strategic user $u$, 
for any possible strategic 
order vector 
${\bf b'}$ of the user $u$, 
either 
${\sf out}^u(x_0, y_0, {\bf b}_{-u}, HS(T)) \succeq_T
{\sf out}^u(x_0, y_0, {\bf b}_{-u}, {\bf b}')$
or 
${\sf out}^u(x_0, y_0, {\bf b}_{-u}, HS(T))$ 
and ${\sf out}^u(x_0, y_0, {\bf b}_{-u}, {\bf b}')$
are incomparable w.r.t. $\preceq_T$.
\label{defn:ic}
\end{definition}

More intuitively, 
the definition says that {\it no strategic play can result
in a strictly better outcome  
than the honest strategy}.

\myparagraph{Incentive compatibility implies arbitrage resilience.}
The fact below shows that incentive compatibility implies
arbitrage resilience. 

\begin{fact}
Suppose that the potential function $\Phi$ is increasing. 
We have that incentive compatibility 
implies arbitrage resilience. 
\label{fct:SPimpliesAR}
\end{fact}
\begin{proof}
\ignore{
Given that $\Phi$ is increasing, 
the market exchange rate $r(x, y) > 0$ for any positive
$x$ and $y$ --- if this is not the case, 
there must exist some $x' > x$ and $y$
such that $\Phi(x, y) = \Phi(x', y)$ or 
some $y' > y$ and $x$
such that $\Phi(x, y') = \Phi(x, y)$, thus contradicting
the fact that $\Phi$ is increasing.  
This means that 
}
We can prove the contra-positive, that is, 
a mechanism that is not arbitrage free 
 cannot be incentive compatible.
Suppose that 
the
mechanism is not arbitrage free.
This means that there exists a 
list of orders $S = \{(t_i, v_i, r_i, \alpha_i)\}_i$ 
such that under honest execution, a subset 
of the orders 
$S' \subseteq S$ 
will enjoy $\delta x \geq 0$ and $\delta y \geq 0$,
and at least one of the two is strictly positive. 
Now imagine that there is a world 
that consists of the orders 
$S \backslash S'$.  
In this case, a strategic user or miner with 0 demand
can  
inject a set of fake orders $S'$, and 
clearly this strategic behavior has positive gain, thus violating
incentive compatibility.
Note that this strategy works even in the weak fair-sequencing model,
as long as the 
strategic user's inherent arrival time $\alpha$ is no larger
than the orders in $S'$.  
\end{proof}

Since our plain-model mechanism 
gives an example that satisfies arbitrage resilience
but not incentive compatible, 
we conclude that 
incentive compatibility is {\it strictly} stronger
than arbitrage resilience.  
Specifically, at the beginning of 
\Cref{sec:refinement-orderfair}, we explain 
why we cannot achieve incentive compatible in the plain model.
The strategies presented in this section 
also help to illustrate why the stronger
notion of incentive compatibility 
is more desirable.

\elaine{TODO: associate line numbers with phases in narrative}

\begin{figure*}[p]
\begin{mdframed}
    \begin{center}
    {\bf Our swap mechanism} 
    \end{center}

\vspace{2pt}
\noindent
    \textbf{Input:} 
A current pool state ${\sf Pool}(x_0, y_0)$, 
and a vector of orders 
${\bf b}$. 
Since the mechanism 
does not make use of the auxiliary information field,
we simply assume each order is a tuple of the form $(t, \amt, r)$.
    
\vspace{5pt}
\noindent\textbf{Mechanism:}
\begin{enumerate} 
\item 
Let 
$r_0 := r(x_0, y_0)$ be the initial exchange rate.
Ignore all ${\sf Buy}(X)/{\sf Sell}(Y)$
orders whose specified rate 
$r < r_0$, and ignore all ${\sf Buy}(Y)/{\sf Sell}(X)$ orders 
whose specified rate $r > r_0$.
Let 
${\bf b}'$ 
be the remaining orders.
\item 
Let $\sigma = \sum_{(t, \amt, r) \in {\bf b}'} \beta(t, \amt, r)$
where $\beta(t, \amt, r) = \begin{cases}
\amt & \text{ if } t = {\sf Buy}(X)\\
-\amt & \text{ if } t = {\sf Sell}(X)\\
-\amt/r_0 & \text{ if } t = {\sf Buy}(Y)\\
\amt/r_0 & \text{ if } t = {\sf Sell}(Y)
\end{cases}$

We call $\sigma \geq 0$ the ${\sf Buy}(X)/{\sf Sell}(Y)$-dominant case,
and $\sigma < 0$ the ${\sf Buy}(Y)/{\sf Sell}(X)$-dominant case.
\item 
The ${\sf Buy}(X)/{\sf Sell}(Y)$-dominant case (if $\sigma \geq 0$): 
\begin{enumerate}[leftmargin=5mm,itemsep=1pt]
\item 
Sort ${\bf b}'$ such that all the 
${\sf Buy}(Y)/{\sf Sell}(X)$
orders 
appear in front of the 
${\sf Buy}(X)/{\sf Sell}(Y)$ orders. 
Write the resulting list 
of orders as $\{(t_i, \amt_i, r_i)\}_{i \in [n']}$.
\label{item:sort1}
\item 
Henceforth we assume that there exists an index 
$j \in [n']$ such that 
$\sum_{i = 1}^j \beta(t_i, \amt_i, r_i) = 0$.
If not, 
we can find 
the smallest index 
$j \in [n']$ 
such that $\sum_{i = 1}^j \beta(t_i, \amt_i, r_i) > 0$, 
and split the $j$-th order into two orders 
$(t_j, \amt_{j, 0}, r_j)$
and $(t_j, \amt_{j, 1} , r_j)$, resulting in a new list with $n'+1$
orders, such that  $\amt_{j, 0} + \amt_{j, 1} = \amt_j$, 
and moreover, 
index $j$ of the new list satisfies 
this condition.
\label{item:sort1b}
\item 
\label{item:phase1}
Phase 1: 
Fully execute the first $j$ orders 
at the initial rate $r_0$.
\item 
\label{item:exec1}
Phase 2: 
For each $i \geq j + 1$ in sequence, 
fulfill as much of the $i$-th remaining order as possible, that is, 
pick the largest $\amt \leq \amt_i$ such that  
subject to the constant-function market maker $\potential$, 
the new market rate 
$r \leq r_i$  
if $\amt$ units are to be executed; 
execute $\amt$ units of the $i$-th order.
\end{enumerate}

\item 
The ${\sf Buy}(Y)/{\sf Sell}(X)$-dominant case is symmetric
(if $\sigma < 0$): 
\begin{enumerate} 
\item 
Sort ${\bf b}'$ such that all the 
${\sf Buy}(X)/{\sf Sell}(Y)$ orders 
appear {\color{red}after} the 
${\sf Buy}(Y)/{\sf Sell}(X)$ orders. 
Write the resulting list 
of orders as $\{(t_i, \amt_i, r_i)\}_{i \in [n']}$.
\label{item:sort2}
\item 
Henceforth we assume that there exists an index 
$j \in [n']$ such that 
$\sum_{i = 1}^j \beta(t_i, \amt_i, r_i) = 0$.
If not, we can find 
the smallest index 
$j \in [n']$ 
such that ${\color{red}-} \sum_{i = 1}^j \beta(t_i, \amt_i, r_i) > 0$, 
and split the $j$-th order into two orders 
$(t_j, \amt_{j, 0}, r_j)$
and $(t_j, \amt_{j, 1} , r_j)$, resulting in a new list with $n'+1$
orders, such that  $\amt_{j, 0} + \amt_{j, 1} = \amt_j$, 
and moreover, 
index $j$ of the new list satisfies 
this condition.
\label{item:sort2b}
\item 
\label{item:phase1'}
Phase 1: 
Fully execute the first $j$ orders at the initial rate $r_0$.
\item 
\label{item:exec2}
Phase 2: 
For $i \geq j + 1$ in sequence, 
fulfill as much of the $i$-th remaining order as possible,
that is, 
pick the largest $\amt \leq \amt_i$ such that  
subject to the constant-function market maker $\potential$, 
the new market exchange rate 
$r {\color{red} \geq } r_i$ 
if $\amt$ units are to be executed;  
execute $\amt$ units of the $i$-th order.
\end{enumerate}
\end{enumerate}
\end{mdframed}
\caption{Our swap mechanism}
\label{fig:mech}
\end{figure*}

\section{Our Swap Mechanism}

\subsection{Construction}
Our swap mechanism has two phases. 
In phase 1 (line \ref{item:phase1} and \ref{item:phase1'}), 
the mechanism matches ${\sf Buy}(X)$
orders with ${\sf Buy}(Y)$
orders, and (partially) executes 
them at the initial rate $r_0$, 
such that at the end, there is no change to the initial pool state 
${\sf Pool}(x_0, y_0)$.
Phase 2 (line \ref{item:exec1} and \ref{item:exec2}) 
is a ${\sf Buy}(X)$-only
phase, in which a sequence of 
${\sf Buy}(X)$
orders (or 
${\sf Buy}(Y)$ orders)
are (partially) executed one by one.
In phase 2, when the mechanism attempts to execute
an order, it will execute as much as possible
until either the demand has been fulfilled, or  
the new market price has reached the asking price.
The details of the mechanism 
are described in \Cref{fig:mech}.


In the sorting steps (Lines~\ref{item:sort1} and \ref{item:sort2}\footnote{Note that the sorted order produced by (\ref{item:sort1})
will be consumed in the subsequent steps
(\ref{item:sort1b}), (\ref{item:phase1}), and (\ref{item:exec1}); 
similarly,
the sorted order of (\ref{item:sort2})
are consumed in (\ref{item:sort2b}), (\ref{item:phase1'}), 
and (\ref{item:exec2}).}),
we may need to break ties 
among identical orders.
We suggest two 
approaches for tie-breaking:
\begin{itemize} 
\item 
In the presence of a centralized block proposer (i.e., 
when the arbitrager can 
 be in full control of block creation), 
we suggest {\it random tie-breaking}. 
Note that the random tie-breaking is 
enforced by the mechanism (i.e., the smart contract).
The random coins needed should come from 
a trusted source at the consensus layer,
e.g., through the use of a fair coin toss protocol~\cite{randombeacon00,randombeacon01}.
This 
ensures that the miner or block producer
cannot auction off favorable positions in the block 
to arbitragers or users. 

{Note that in the plain model, we aim to achieve only arbitrage-free 
and not incentive compatibility. 
So although it may seem like a user can adopt a Sybil strategy  
and submit fake bids to game the random tie-breaking,
such strategies do not actually violate  
the arbitrage-free  
property.}. 
\item 
If the block proposal process is decentralized
and ensures sequencing fairness, we suggest
tie-breaking 
according to the time of arrival. \Cref{sec:refinement-orderfair} 
and \Cref{sec:proof-weakic}
show that this approach allows 
us to achieve incentive compatibility in the 
weak fair-sequencing model.
\end{itemize}

\subsection{Proof of Arbitrage Resilience}

We now prove that our swap mechanism satisfies arbitrage resilience
regardless of how ties are broken in 
Lines~\ref{item:sort1} and \ref{item:sort2}. 
As mentioned earlier, the arbitrage resilience
property holds even when the arbitrager 
is in full control of block creation, can 
drop or inject orders,  
and can control the sequencing of orders within the block.

\begin{theorem}[Arbitrage resilience] 
\label{thm:strong-arbitrage}
The swap mechanism in \Cref{fig:mech} satisfies arbitrage resilience.
In particular, this holds no matter how ties are broken
in Lines~\ref{item:sort1} and \ref{item:sort2}. 
\end{theorem}
\begin{proof}
We prove it for the 
${\sf Buy}(X)/{\sf Sell}(Y)$-dominant case, 
since the 
${\sf Buy}(Y)/{\sf Sell}(X)$-dominant case
is symmetric.
The mechanism essentially does the following. In phase 1, 
it (partially) executes a set of orders all at the 
initial rate $r_0$, such that there is no change to the initial pool state 
${\sf Pool}(x_0, y_0)$. In phase 2, it executes only 
${\sf Buy}(X)/{\sf Sell}(Y)$ orders. Due to increasing marginal cost (\Cref{fact:inc-marginal-cost}), 
\elaine{TODO: state the lemma earlier}
in Phase 2, all the (partially) executed ${\sf Buy}(X)/{\sf Sell}(Y)$ orders 
enjoy a rate that is at least $r_0$. 
Therefore, all the (partially) executed ${\sf Buy}(Y)/{\sf Sell}(X)$
orders 
enjoy a rate of $r_0$, and all 
the (partially) executed
${\sf Buy}(X)/{\sf Sell}(Y)$ orders enjoy a rate that is $r_0$ or greater.
Thus, it cannot be the case that 
there is a net gain in one asset without any loss in the other.
\ignore{
\elaine{FILL}
Fix an arbitrary initial state ${\sf Pool}(x_0,y_0)$ and an arbitrary order vector from all other users ${\bf b}_{-u}$.
Consider an arbitrary set of orders ${\bf b}_{u}$ whose joint outcomes result in $\delta x \geq 0$.
Without loss of generality, assume that 
we have the ${\sf Buy}(X)$/${\sf Sell}(Y)$-dominant case 
given ${\bf b}_{-u}$ and ${\bf b}_u$ (the 
${\sf Buy}(Y)$/${\sf Sell}(X)$-dominant case
is symmetric).
\myparagraph{Case 1: $\delta x = 0$.} If no order in ${\bf b}_u$ is executed, then the joint outcomes result in $\delta y = 0$.
Thus, we assume that all the ${\sf Sell}(X)$/${\sf Buy}(Y)$-type orders in ${\bf b}_u$ has a joint outcome $(\delta x_s, \delta y_s)$, whereas all the ${\sf Buy}(X)$/${\sf Sell}(Y)$-type orders in ${\bf b}_u$ has a joint outcome $(\delta x_b, \delta y_b)$.
}
\end{proof}

\subsection{A Refinement of the Mechanism for the Weak Fair-Sequencing Model}
\label{sec:refinement-orderfair}

\myparagraph{Why the mechanism is NOT strategy-proof in the plain model.}
As mentioned, our swap mechanism 
in \Cref{fig:mech} 
does not fully specify 
how to break ties in the sorting steps of  
in Lines~\ref{item:sort1} and \ref{item:sort2}. 
Moreover, our arbitrage resilience property
does not care how the tie-breaking is done. 

However, 
if we are not careful about the tie-breaking
the resulting mechanism may not satisfy incentive compatibility. 
For example, imagine that 
the tie-breaking is based on the $\alpha$ field
of the order, which encodes the time at which the order
is submitted.  
In this case, a strategic miner or user  
can simply 
make its own order have a small $\alpha$ to enjoy
a better price, rather than truthfully reporting 
$\alpha$.

\myparagraph{Refinement in the fair-sequencing model.}
In the weak fair-sequencing model, 
the arrival order $\alpha$ is decided by an underlying
consensus protocol 
that provides some form of sequencing fairness.
In other words, the consensus protocol itself
records the approximate time at which each order is first seen. 
In practice, this arrival time is dependent
on when the user submits the 
its order to the network, and its network delay. 
A sequencing-fair consensus protocol
cannot prevent a user (or miner) from delaying
the submission of its order. 
It also cannot prevent a strategic 
user with an extremely fast network from front-running
a victim's order. 
Specifically, the strategic user can still observe
what the victim 
submits, and then immediately submits a dependent order. 
If the strategic user's network is faster than the victim's,
the strategic order may have an earlier arrival time than the victim's 
order!
However,
sequencing fairness 
from the consensus does tie the hands 
of strategic players in the following weak
manner (hence the name ``weak sequencing fair''):
if a strategic user's order is generated at time $\alpha$, it cannot
pretend that the order was generated at $\alpha' < \alpha$
even if it has an extremely fast network with a delay of $0$. 
\elaine{move earlier?}

We can refine the mechanism in 
\Cref{fig:mech}
as follows to achieve incentive compatibility 
in the weak fair-sequencing model  
(see also \Cref{sec:proof-weakic}):
\begin{itemize} 
\item[]
For tie-breaking in the sorting steps
(Lines~\ref{item:sort1} and \ref{item:sort2}), we now require
that 
the sorting 
be stable, that is, in the sorted outcome, 
the relative ordering among all 
identical ${\sf Buy}(Y)$
orders must respect the arrival order encoded
in the $\alpha$ field; 
the same holds for all 
${\sf Buy}(X)$ orders. 
\end{itemize}

We will assume
that a user's honest strategy is to honestly report its type
including the $\alpha$ field, that is,
$HS(t, \amt, r, \aux)$
simply outputs 
a single order 
$(t, \amt, r, \aux)$.

\begin{theorem}[Incentive compatibility in the weak fair-sequencing model] 
Suppose $\potential$ is concave, increasing, and differentiable.
In the weak fair-sequencing model, 
the above refined swap mechanism 
is incentive compatible (see \Cref{defn:ic}).
\label{thm:weakic}
\end{theorem}

The proof of \Cref{thm:weakic}
is 
provided in \Cref{sec:proof-weakic}.

\section{Proof of Incentive Compatibility in the Weak Fair-Sequencing Model}
\label{sec:proof-weakic}

\elaine{TODO: remove the Sell}

\subsection{Useful Facts}
We first prove a few useful facts.

\begin{fact}[Increasing marginal cost]
Suppose that $\potential$ is 
increasing, differentiable, and concave. 
Given two pool states ${\sf Pool}(x, y), {\sf Pool}(x', y')$
such that $\potential(x, y) = \potential (x', y')$, and $x' \leq x$,
it must be that 
$r(x, y) \leq r(x', y')$.
In other words, the price of $X$ goes up 
if the pool has less supply of $X$.
\label{fact:inc-marginal-cost}
\end{fact}
\begin{proof}
    Suppose $\potential(x,y) = C$.
    Since $\potential$ is increasing,
    by Lemma B.1 of \cite{credible-ex},
    the potential function $\potential$ defines a bijective decreasing function $h(\cdot)$ such that $\potential(z,h(z)) = C$.
    Moreover, since $\potential$ is concave, the induced function $h(\cdot)$ is convex (Lemma B.2 of \cite{credible-ex}). 
    Observe that $\Phi$ is differentiable, so $r(x,y) = -h'(x)$. 
    Therefore, $r(x,y) \leq r(x',y')$ for $x'\leq x$ by the convexity of $h$.
\end{proof}

\begin{fact}[No free lunch]
\label{fact:no_free_lunch}
Suppose $(\delta x, \delta y)$ is the outcome resulting from
the execution of a single order in the swap mechanism in \Cref{fig:mech}.
If at least one of $\delta x$ and $\delta y$ is non-zero,
then $\delta x \cdot \delta y < 0$.
\end{fact}

\subsection{Proof}

We now prove 
\Cref{thm:weakic}.
Our mechanism is deterministic, so a deterministic strategy
yields a deterministic outcome. 
Recall that a user has a partial ordering among the outcomes.
Henceforth, if 
two outcomes satisfy $(\delta x_1, \delta y_1) \succeq (\delta x_2, \delta y_2)$,
we say that 
$(\delta x_1, \delta y_1)$ is {\it at least as good as} $(\delta x_2, \delta y_2)$.

Suppose that the strategic user $u$'s type is $(t^*, \amt^*, r^*, \aux^*)$, 
its strategic order vector is 
${\bf b}_u :=  \{(t_j, \amt_j, r_j, \aux_j)\}_j$, 
the initial state is ${\sf Pool}(x_0, y_0)$, 
and 
the order vector from all other users is ${\bf b}_{-u}$.

\begin{fact}
\label{fact:prime_asset}
    Given the initial state ${\sf Pool}(x_0, y_0)$, the order vector from other users ${\bf b}_{-u}$, and a strategic order vector ${\bf b}_u :=  \{(t_j, \amt_j, r_j, \aux_j)\}_j$,
    there exists an alternative strategic vector ${\bf b}'_u$ which contains only ${\sf Buy}(X)$ and ${\sf Sell}(X)$-type (or ${\sf Buy}(Y)$ and ${\sf Sell}(Y)$-type) order, such that the outcome of ${\bf b}'_u$ is the same as the outcome of ${\bf b}_u$.
\end{fact}
\begin{proof}
    Since the mechanism is deterministic, 
    given the initial state ${\sf Pool}(x_0,y_0)$, order vector from all other users ${\bf b}_{-u}$, and the strategic order vector ${\bf b}_u$,
    one can compute the whole order execution trace of the mechanism.
    For an order $b = ({\sf Sell}(Y), \amt_j, r_j, \aux_j)$ in ${\bf b}_u$, if it is not executed, then it can be replaced with a $({\sf Buy}(X), 0, r_j, \alpha_j)$.
    If it is partially fulfilled,
    let $(x_{\rm start}, x_{\rm end})$ denote the amount of asset $X$ in the pool right before and right after $b$ is executed, respectively.
    Let $r_{\rm end}$ denote the market exchange rate at $x_{\rm end}$.
    Then $b$ can be replaced with an order $({\sf Buy}(X), x_{\rm start}- x_{\rm end}, r_{\rm end}, \alpha_j)$, without changing the outcome. 
    Similarly, we can replace a ${\sf Buy}(Y)$-type order with a ${\sf Sell}(X)$-type order without changing the outcome.
    The fact thus follows. 
\end{proof}

\begin{lemma}
Suppose $\potential$ is concave, increasing, and differentiable.
For any strategic order vector ${\bf b}_u$, there exists
a single order $b'_u$ 
such that 
1) $b'_u$ results in an outcome at least as good as 
${\bf b}_u$; 
2) the arrival time used in $b'_u$  
is no earlier than the earliest arrival time in ${\bf b}_u$;
and 3) 
either $b'_u = (\_, 0, \_, \_)$ 
or 
$b'_u$ would be completely 
fulfilled 
under ${\sf Pool}(x_0, y_0)$ and ${\bf b}_{-u}$.
\label{lem:coalesce}
\end{lemma}
\begin{proof}

To prove this lemma, we first show that we can coalesce
all the ${\sf Buy}(X)/{\sf Sell}(Y)$-type orders into one,
and all the ${\sf Sell}(X)/{\sf Buy}(Y)$-type
orders into one, as stated in the following claim.

\begin{claim}
\label{claim:coalesce}
Suppose $\potential$ is concave, increasing, and differentiable.
For any strategic order vector ${\bf b}_u$, 
if it contains both ${\sf Buy}(X)/{\sf Sell}(Y)$-type
and ${\sf Sell}(X)/{\sf Buy}(Y)$-type
orders, 
there exists
another order vector ${\bf b}'_u$ which
contains a single 
${\sf Buy}(X)/{\sf Sell}(Y)$-type
order 
and 
a single 
${\sf Sell}(X)/{\sf Buy}(Y)$-type
order 
such that  
1) 
${\bf b}'_u$ results in an outcome at least as good as ${\bf b}_u$;
2) the arrival times used in ${\bf b'}_u$  
are no earlier than the earliest arrival time in ${\bf b}_u$; 
and 3) an order in ${\bf b}'_u$ is either of the form $(\_, 0, \_, \_)$ 
or it would be completely 
fulfilled 
under ${\sf Pool}(x_0, y_0)$ and ${\bf b}_{-u}$.
\end{claim}
\begin{proof}
According to \cref{fact:prime_asset}, we can assume that ${\bf b}_u$ contains only ${\sf Buy}(X)$ and ${\sf Sell}(X)$-type orders.
Let ${\bf b}_{\sf Sell}$ and ${\bf b}_{\sf Buy}$ denote the vector 
of ${\sf Sell}(X)$-type and ${\sf Buy}(X)$-type orders in ${\bf b}_u$, respectively.
Without loss of generality, assume that given ${\bf b}_{-u}$ and ${\bf b}_{u}$, we have 
the ${\sf Buy}(X)$/${\sf Sell}(Y)$-dominant case, 
(the ${\sf Buy}(Y)$/${\sf Sell}(X)$-dominant case is symmetric).

For the ${\sf Buy}(X)$/${\sf Sell}(Y)$-dominant case, 
all ${\sf Sell}(X)$-type orders in ${\bf b}_{\sf Sell}$ will be executed at the initial exchange rate $r_0$.
Therefore, consider an order $b = ({\sf Sell}(X), \amt, r, \alpha)$, where $\amt$ denotes the total units of orders in ${\bf b}_{\rm Sell}$, $r$ is the minimum asking rate in ${\bf b}_{\rm Sell}$, and $\alpha$ is the earliest arrival time in ${\bf b}_u$.
Then a strategic order vector $({\bf b}_{\rm Buy}, b)$ results in the same outcome as ${\bf b}_u$.

Now consider an order $b' = ({\sf Buy}(X), \amt', r', \alpha')$, where $\amt'$ is the total units of order executed in ${\bf b}_{\rm Buy}$, $r'$ is the maximum asking rate in ${\bf b}_{\rm Buy}$,  and $\alpha'$ is the earliest arrival time of the order in ${\bf b}_{\rm Buy}$ that is partially fulfilled.
Let ${\bf b}'_u = (b', b)$. 
Compared to the units executed in ${\bf b}_{\rm Buy}$, the units executed in $b'$ have an earlier or the same arrival time.
\ke{ref to \Cref{lem:earlygood}?}
In addition, the asking rate $r'$ in $b'$ is larger than or equal to that in ${\bf b}_{\rm Buy}$. 
Therefore, all units in $b'$ will be executed.
Moreover, by the increasing marginal cost (\cref{fact:inc-marginal-cost}), the exchange rate for $b'$ is no more than the average exchange rate for all orders in ${\bf b}_{\rm Buy}$.
This means that ${\bf b}'_u$ results in an outcome that is at least as good as $({\bf b}_{\rm Buy}, b)$ according to the partial ordering.
By the transitivity, ${\bf b}'_u$ results in an outcome that is at least as good as ${\bf b}_u$.
\end{proof}

Next, we show that for any order vector ${\bf b}_u$ that contains a single 
${\sf Buy}(X)/{\sf Sell}(Y)$-type
order
and
a single
${\sf Sell}(X)/{\sf Buy}(Y)$-type
order, we can remove the part that ``cancels off'', and 
substitute it with a single order. This is formally stated
in the following claim.

\begin{claim}
\label{claim:one}
Suppose $\potential$ is concave, increasing, and differentiable.
For any order vector ${\bf b}_u$ that contains a single 
${\sf Buy}(X)/{\sf Sell}(Y)$-type order
and a single ${\sf Sell}(X)/{\sf Buy}(Y)$-type, 
there exists a single order $b'_u$  
such that 
1) $b'_u$ results in an outcome at least as good as ${\bf b}_u$;
2) 
the arrival time in 
$b'_u$  
is no earlier than the earliest arrival time in  
${\bf b}_u$; 
and 3) ${b}'_u$ is either of the form $(\_, 0, \_, \_)$ 
or it would be completely 
fulfilled 
under ${\sf Pool}(x_0, y_0)$ and ${\bf b}_{-u}$.
\end{claim}

\begin{proof}
Because of \cref{fact:prime_asset},
we assume that 
the strategic order vector ${\bf b}_u$ contains a single 
${\sf Buy}(X)$-type order $({\sf Buy}(X), \amt_b, r_b, \alpha_b)$ and a single 
${\sf Sell}(X)$-type order 
$({\sf Sell}(X), \amt_s, r_s, \alpha_s)$. By our assumption,  
both orders are fully executed.
If either $\amt_b$ or $\amt_s$ is zero, 
the result follows trivially;  
henceforth, we assume that both are non-zero.
Similarly, if either $r_b < r_0$ or 
$r_s > r_0$, then the result also follows trivially. 
henceforth, 
$r_b \geq r_0$ and $r_s \leq  r_0$.

We prove it assuming the 
${\sf Buy}(X)$/${\sf Sell}(Y)$-dominant case
case 
under the strategic order
vector ${\bf b}_u$, 
since the argument for the 
${\sf Buy}(Y)$/${\sf Sell}(X)$-dominant case
 is symmetric.
For the ${\sf Buy}(X)$/${\sf Sell}(Y)$-dominant case, 
if it were possible to execute all orders at $r_0$,
then there is more demand in terms of  
only ${\sf Buy}(X)/{\sf Sell}(Y)$ than 
${\sf Sell}(X)/{\sf Buy}(Y)$. 
In this case, 
Phase 1 executes 
all ${\sf Sell}(X)/{\sf Buy}(Y)$ orders at $r_0$, and  
Phase 2 executes 
only ${\sf Buy}(X)/{\sf Sell}(Y)$ orders. 

\myparagraph{Case $\amt_s \geq \amt_b$:}
Under the original ${\bf b}_u$, 
the user would sell $\amt_s$ units of $X$ at $r_0$
and would buy $\amt_b$ units of $X$ at a rate of $r_0$ or greater.
Now, suppose we replace ${\bf b}_u$ with a single order 
$b'_u = ({\sf Sell}(X), \amt_s - \amt_b, r_s, \alpha_s)$.
Under $b'_u$, it is still the case that all  
${\sf Sell}(X)$ orders are completely executed at $r_0$.

Hence, we can decompose the original
${\bf b}_u$ equivalently into the following steps:
(i) first execute $b'_u $; (ii) sell $\amt_b$ units at rate $r_0$;
(iii) buy back $\amt_b$ units at rate $r_0$ or greater.
Since steps (ii) and (iii) together will incur a non-negative
loss in $Y$ (but create no change in $X$),
user $u$'s outcome under $b'_u$ 
is at least as good as ${\bf b}_u$.

\myparagraph{Case $\amt_s < \amt_b$:}
Consider the original ${\bf b}_u$
which consists of 
$({\sf Sell}(X), \amt_s, r_s, \alpha_s)$
and $({\sf Buy}(X), \amt_b, r_b, \alpha_b)$.
We will analyze what happens when we replace these two orders with 
a single order $b'_u = ({\sf Buy}(X), \amt_b - \amt_s, r_b, \alpha_b)$.
Suppose in Phase 1 of the execution with ${\bf b}_u$, user $u$ sells $\amt_s$ units of $X$ 
and buys $\amt'\leq \amt_b$ units of $X$ at rate $r_0$.  We separate the rest of the proof into two cases.
\begin{enumerate} 

\item \textbf{Case $\amt' \geq \amt_s$.} 
For the execution with ${\bf b}_u$:
In Phase 1, the net effect is to buy $\amt' - \amt_s$
units of $X$ at rate $r_0$.
In Phase 2, the user~$u$ buys the remaining
$\amt_b - \amt'$ units starting at rate $r_0$.

The execution with $b'_u$ can be viewed as follows:
In Phase 1', $\amt'-\amt_s$ units of ${\sf Buy}(X)$ will be executed at rate $r_0$.
Then in Phase 2', the mechanism executes the rest $\amt_b - \amt'$ units in $b'_u$ starting at rate $r_0$.

Hence, the two scenarios are equivalent, and the two outcomes are the same.

\item \textbf{Case $\amt' < \amt_s$.} 
We will view the execution of the orders
$({\sf Sell}(X), \amt_s, r_s, \alpha_s)$
and $({\sf Buy}(X), \amt_b, r_b, \alpha_b)$
as follows.

\begin{itemize} 
\item 
{\it Phase 1: } User $u$ sells $\amt_s$ units of $X$ and buy $\amt'$ units of $X$ at a rate of $r_0$,
and it gains $(\amt_s-\amt') \cdot r_0$ units of $Y$ in return.

\item 
{\it Phase 2a:}
Some non-negative amount 
of ${\sf Buy}(X)$/${\sf Sell}(Y)$
orders from other users are executed at a starting market rate of $r_0$,
let $\amt_{\rm other} \geq 0$
be the units of $X$ purchased. 
Note that if $\amt' > 0$, then $\amt_{\rm other} = 0$.

At the end of this phase, the market rate $r_1 \geq r_0$ by
increasing marginal cost.
\item 
{\it Phase 2b:}
Starting at rate $r_1 \geq r_0$, user~$u$ buys
$\amt_s -\amt'$ units of $X$, which changes the market rate to $r_2 \geq r_1$.
\item 
{\it Phase 2c:}
Starting at rate $r_2$,
the user buys $(\amt_b - \amt_s)$ units of $X$,
changing the market rate to $r_3$.
\end{itemize}

The new execution involving  $b'_u$ can be viewed as the following:
\begin{itemize} 
\item 
{\it Phase 1'}: The ${\sf Buy}(X)/{\sf Sell}(Y)$ orders of other users 
executed in the original Phase 1
cannot all be executed in the new Phase 1. 
In particular, the last $\amt_s - \amt'$ units of $X$ cannot be fulfilled in Phase 1',
and will be pushed to Phase 2a' --- henceforth, we call this portion the residual.

\item 
{\it Phase 2a'}: The mechanism attempts to execute 
the residual from Phase 1'
at a starting rate of $r_0$.
The amount fulfilled must be at most $\amt_s-\amt'$.
\item 
{\it Phase 2b'}: 
The mechanism attempts to execute the (partial) orders 
originally considered 
in Phase 2a, at a starting price that is at least $r_0$.
At most $\amt_{\rm other}$ units of $X$
can be fulfilled.
The ending market rate must be at least $r_1$.
\item 
{\it Phase 2c'}: 
The mechanism attempts to execute $b'_u$.
Observe
that the total units of $X$ fulfilled
in the original Phase 2a, 2b, and 2c is $\amt_{\rm other} + \amt_b - \amt'$,
and the total units of $X$ fulfilled in the new
Phase 2a', 2b', and 2c' is at most $(\amt_s - \amt') + \amt_{\rm other} + \amt_b - \amt_s = 
\amt_{\rm other} + \amt_b - \amt'$.
Therefore, it must be that 
all of $b'_u$ can be fulfilled and the ending market rate
is at most $r_3 \leq r_b$.
\end{itemize}

Henceforth, we use the notation ${\sf Pay}(\text{Phase *})$
to denote  user $u$'s payment in terms of $Y$ in some phase.
When the pay is negative, it means a gain in $Y$.
Observe that 
\[
{\sf Pay}(\text{Phase 1}) + {\sf Pay}(\text{Phase 2b})  \geq 0, 
\quad
{\sf Pay}(\text{Phase 2c}) 
\geq 
{\sf Pay}(\text{Phase 2c'}).
\]
Therefore, 
\[
{\sf Pay}(\text{Phase 1}) + {\sf Pay}(\text{Phase 2b})  
+ {\sf Pay}(\text{Phase 2c}) 
\geq 
{\sf Pay}(\text{Phase 2c'}).
\]
Observe that in the above, the left-hand side
represents $u$'s total payment in $Y$ under the original ${\bf b}'_u$,
and 
the right-hand 
side 
represents $u$'s total payment in $Y$ under the new $b'_u$.
\end{enumerate}
\end{proof}

\Cref{lem:coalesce} follows by combining \Cref{claim:coalesce} and \Cref{claim:one}.
\end{proof}

\begin{lemma}
Suppose $\potential$ is concave, increasing, and differentiable.
Given any 
initial state ${\sf Pool}(x_0, y_0)$, 
any order vector ${\bf b}_{-u}$, 
any true arrival time $\aux^*$ of user $u$, 
given an order $b'_u$ with an arrival
time later than $\aux^*$, there exists
another order $b_u$ 
with an arrival time exactly $\aux^*$, and moreover, 
user $u$'s outcome under $b_u$ 
is at least
as good as its outcome under $b'_u$, 
and $b_u$ 
is either completely executed or of the form $(\_, 0, \_, \_)$.
\label{lem:earlygood}
\end{lemma}
\begin{proof}
Due to \Cref{fact:prime_asset}, we may assume
that $b_u$ is  
either a ${\sf Buy}(X)$ or ${\sf Sell}(X)$ order. 
We prove it for a ${\sf Buy}(X)$ order, since
the case for a ${\sf Sell}(X)$ order 
is symmetric. 
\ignore{
Let $\delta x$ be the change in $X$ 
under both $b_u$ and $b'_u$. We may assume that $\delta x > 0$
since the case $\delta x = 0$ is trivial. 
We can consider the following cases. 
}
Let $\amt$ be the amount of $X$ bought  
by $b'_u$. We shall assume that $\amt > 0$
since the case $\amt = 0$ is trivial.
Consider an order $b_u = ({\sf Buy}(X), \amt, +\infty, \aux^*)$.
Clearly, $b_u$ will buy $\amt$ units of $X$.

We consider the following cases: 
\begin{itemize} 
\item $b_u$ buys  
$\amt_0 > 0$ units at $r_0$ in Phase 1, and then buys
$\amt_1 \geq 0$ units in Phase 2 at a starting rate of $r_0$ 
and an ending rate of $r_1 \geq r_0$.
In this case, by delaying the arrival time, $b'_u$
can buy at most $\amt'_0 \leq \amt_0$ units in Phase 1 at $r_0$,  
and it needs to buy $\amt - \amt'_0 \geq \amt_1$ in Phase 2. 
Therefore, for $b'_u$, the starting rate in Phase 2 
is $r_0$, and the ending
rate must be at least $r_1$.
Therefore, the average price paid per unit in $b_u$ is no worse
than the
average price paid 
per unit in $b'_u$.
\item 
$b_u$ buys all $\amt$ units in Phase 2.
Since $b'_u$ delayed the arrival, 
the order can be considered no earlier 
than $b_u$. Thus, 
before the mechanism tries to execute $b'_u$ 
at least as many units of $X$ will have been bought
(by all users)
as when the mechanism tries to execute $b_u$.
This means $b'_u$ will have an average price no better than $b_u$. 
\end{itemize}
\end{proof}

Due to \Cref{lem:coalesce,lem:earlygood}, it suffices
to consider strategies that submit a single order,
declare the true arrival 
time $\aux^*$, 
and moreover, either the order has a $0$ amount or
it will be completely executed under ${\sf Pool}(x_0, y_0)$ and ${\bf b}_{-u}$
--- henceforth, we call such strategies 
as {\it admissible, single-order} strategies.
We can complete the proof of \Cref{thm:weakic}
by showing the following lemma.

\begin{lemma}
\label{lemma:single_strategic}
Suppose $\potential$ is concave, increasing, and differentiable.
For any admissible and single-order strategy~$S$, 
the honest strategy results in an outcome that is at least
as good as or incomparable to strategy~$S$. 
\end{lemma}

\begin{proof}

Below we complete the proof of \Cref{lemma:single_strategic}.

We prove it for the case when user $u$'s type is either 
$({\sf Buy}(X), \amt^*, r^*, \aux^*)$ 
or $({\sf Sell}(X), \amt^*, r^*, \aux^*)$. The case
for ${\sf Buy}(Y)$/${\sf Sell}(Y)$ is symmetric.
Given two outcomes 
${\sf out}_0 = (\delta x_0, \delta y_0)$
and ${\sf out}_1 = (\delta x_1, \delta y_1)$ and a true demand $\delta x^*$ for $X$, 
we say that 
they are {\it on the same side} of the goal
$\delta x^*$ iff 
$(\delta x_0 - \delta x^*) \cdot (\delta x_1 - \delta x^*) \geq 0$.
We say that 
${\sf out}_0$ 
is {\it at least as close as}
${\sf out}_1$ 
towards the goal 
iff $|\delta x_0 - \delta x^*| \leq |\delta x_1 - \delta x^*|$,
and we say that 
${\sf out}_0$ 
is {\it closer}
to the goal than 
${\sf out}_1$ 
iff $|\delta x_0 - \delta x^*| < |\delta x_1 - \delta x^*|$.


Our natural partial ordering relation implies the following:
\begin{enumerate} 
\item[R1.]
Suppose $(\delta x_0, \delta y_0)$ 
and $(\delta x_1, \delta y_1)$ 
are on the same side of the goal,
and 
$\delta x_0$ is at least as close as $\delta x_1$ towards the goal. 
Moreover, if $(\delta x_0 - \delta x_1) \cdot (\delta y_1 - \delta y_0) < 0$,
then
$(\delta x_1, \delta y_1) \nsucceq (\delta x_0, \delta y_0)$.

\item[R2.] 
If $\delta x_0$
and 
$\delta x_1$ are on the same side of the goal,
$\delta x_0$ is closer than $\delta x_1$ to the goal, 
and moreover, 
$\delta y_1 - \delta y_0 > r^* \cdot (\delta x_0 - \delta x_1)$, then 
$(\delta x_0, \delta y_0) \nsucceq (\delta x_1, \delta y_1)$.
\elaine{is it just incomparable here?}
\item[R3.]
If $\delta x_0$ and $\delta x_1$ are on different sides of the goal,  
and moreover, 
$\delta y_1 - \delta y_0 \geq r^* \cdot (\delta x_0 - \delta x_1)$, then 
$(\delta x_0, \delta y_0) \nsucceq (\delta x_1, \delta y_1)$.
\elaine{is it just incomparable here?}
\end{enumerate}

Due to \Cref{fact:prime_asset}, we may assume
that the strategic order must be 
of the type 
${\sf Buy}(X)$ or ${\sf Sell}(X)$.  
Further, as shown in the following fact, we can in fact
assume that the strategic order  
adopts the true time of arrival $\aux^*$, i.e., declaring
a later time never helps.

Henceforth, 
let $(\delta x, \delta y)$ and $(\delta x', \delta y')$ denote
the honest and strategic outcomes, respectively.

\myparagraph{Case 1:} Either user $u$ has a true demand of 0 units, 
or the strategic order is opposite the direction of its true demand, i.e.,
if its type is ${\sf Buy}(X)$, it submits a single  
${\sf Sell}(X)$ order; or vice versa. 
It must be that $\delta x \cdot \delta x' \leq 0$.
Further, the honest outcome and strategic outcome
must be on the same side of the true demand, 
and the honest outcome is at least as close 
to the goal as the strategic outcome.
Our mechanism guarantees that 
either (i) $\delta x = \delta x' = 0$, 
or (ii) at least one of $\delta x'$ and $\delta x'$
is non-zero.
In case (i), 
$\delta y = \delta y' = 0$, and the
the two outcomes are the same.

In case (ii),
because of no free lunch (\Cref{fact:no_free_lunch}),
at least one of the inequalities
$\delta x \cdot \delta y \leq 0$ and 
$\delta x' \cdot \delta y' \leq 0$ must be strict; moreover,
when an equality holds, both arguments of the product must be zero.
Because $\delta x \cdot \delta x' \leq 0$,
this implies that
$(\delta x - \delta x')(\delta y - \delta y')  < 0$;
by the above rule R1, the honest outcome is at least as good
as or incomparable the strategic one.

\myparagraph{Case 2:} 
The user $u$ has a non-zero 
amount of true demand, 
and moreover, 
the strategic order 
is in the same direction of the true demand.
We consider the following cases. 
\begin{itemize} 
\item 
{\it Case 2a:} $\delta x = \delta x'$. 
By admissibility, the 
strategic
order declares the same arrival time as the honest one;
hence, if the orders from both strategies get executed
for a non-zero amount, both execution will start at the same market exchange rate.
Hence, it must be the case that 
$(\delta x, \delta y) = (\delta x', \delta y')$.
\elaine{TO FIX: there is an issue of tie breaking if
other orders enjoy the same arrival time}

Henceforth, we assume that $\delta x \neq \delta x'$.
\item 
{\it Case 2b:}
The honest outcome and the strategic outcome are on the same side
of the goal, and 
the honest outcome 
closer to the goal than the strategic outcome;
this case includes the scenario that the honest outcome is exactly at the goal.
Since we can assume that the strategic order
has the same arrival time as the honest order,  
the difference of $|\delta x - \delta x'|$ units are traded at 
a marginal price 
at least as good as $r$ in the honest outcome.
Due to the third rule of the natural partial ordering, 
\elaine{NOTE: hardcoded ref}
the honest outcome is at least as good as the stategic one. 

\item  {\it Case 2c:}
The honest outcome and the strategic outcome are on the same side
of the goal, and the strategic outcome 
is closer to the goal than the honest outcome,
i.e., $|\delta x| < |\delta x'| \leq |\delta x^*| = \amt^*$. 

This means that the honest outcome has not reached the goal of user~$u$.
Under the honest strategy, 
after user $u$'s 
order has been executed (or attempted to be executed), 
the state of the market is such that
if a further non-zero portion of the order is executed,
this portion will incur an average rate of strictly worse than $r^*$.
In the case of ${\sf Buy}(X)$, this is strictly larger than $r^*$;
in the case of ${\sf Sell}(X)$, this is strictly less than $r^*$.

Since the strategic order declares the same arrival
time as the honest order, 
the difference of $|\delta x' - \delta x| > 0$
units must be traded at an average rate strictly worse than 
than $r^*$ 
in the strategic outcome.
\elaine{TODO: need more explanation here.}
By rule R2, 
the strategic outcome is not at least as good
as the honest outcome.
However, because of no free lunch,
the two outcomes are actually incomparable.

\item  {\it Case 2d:}
The honest outcome and the strategic outcome 
are on different sides of the goal,
i.e., $|\delta x| < |\delta x^*| = \amt^* < |\delta x'|$.

In this case, $\delta x \neq 0$.
Similar to case 2c, 
under the honest strategy, 
after user $u$'s 
order has been executed, 
the state of the market is such that
if a further non-zero portion of the order is executed,
this portion will incur an average rate of strictly worse than $r^*$.

Because the 
strategic order declares the same arrival time as the honest one,
the difference of $|\delta x' - \delta x|$
units must be traded at an average rate of strictly worse than
than $r^*$ 
 in the strategic outcome.
By rule R3, 
the strategic outcome cannot be at least as good as the honest outcome. \elaine{is it incomparable?}
\elaine{TODO:}

\end{itemize}

\elaine{what if one extra is fulfilled at exactly asking, 
do we ever need this case}

\elaine{TODO: may need to augment the order with an extra parameter}

\end{proof}
\section{Conclusion}
\label{sec:conclusion}

In this paper, we propose new models
for studying  mechanism design for DeFi applications. 
Unlike the prior work of 
Ferreira and Parkes~\cite{credible-ex} and others'~\cite{greedyseq}
which assume
that the mechanism on the blockchain must be 
a first-in-first-out
mechanism, we allow the the mechanism designer
to  specify the mechanism running on the blockchain. 
This allows us 
to circumvent the strong impossibility results of Ferreira and Parkes~\cite{credible-ex}.
Depending on 
assumption on 
the underlying blockchain, we consider two possible strategies spaces.
If the underlying blockchain does not enjoy sequencing fairness,
we assume that the strategic user (or miner)
can post orders after observing honest users' orders, insert
fake orders, 
censor honest users' orders, 
and control the sequencing of the orders in the block.
If the underlying blockchain enjoys sequencing fairness,
we assume that the strategic user
can do all of the above; however, it cannot censor honest users' orders,
nor can it under-report its arrival time.


We design a novel mechanism 
that achieves arbitrage resilience (which was deemed
impossible under 
Ferreira and Parkes~\cite{credible-ex}'s model),
and additionally achieves incentive compatibility 
if the underlying consensus offers sequencing fairness. 

Our paper  
raises many interesting directions for future work. For example,
can we achieve incentive compatibility without
relying on the sequencing fairness assumption?
Can we extend the results to multi-asset swaps?
Can we optimize social welfare and revenue 
under incentive compatibility?
Another interesting direction is 
whether we can achieve 
incentive compatibility
under an all-or-nothing fulfillment model, that is,
any order is either completely fulfilled
or not executed at all.


\bibliographystyle{alpha}
\bibliography{refs,crypto,bitcoin}

 \appendix
\section*{Appendices}
\section{Full Definition: Partial Ordering of Outcomes}
\label{sec:rank}

We first define a most natural partial ordering among a user's outcomes. 
We do not directly define a total ordering
since some outcomes may not 
be directly comparable without extra information.
Intuitively, suppose that a user's type
is $({\sf Buy}(X), \amt, r, \_)$, 
the natural partial ordering
says that 1) up to receiving at most $\amt$ units of $X$, 
every extra unit of $X$ received at a marginal price
better than $r$ is desirable; 
2) every excessive unit (i.e., more than the $\amt$ amount)
of $X$ received at a marginal price of
more than $r$ is undesirable;
 3) every 
unit of $X$ 
short-sold   
at a price less than $r$ is undesirable;
and 4) if the net gain in $X$
is the same in two outcomes, then
the outcome where a lower price is paid 
is more desirable.

\myparagraph{Example.}
For example, imagine that a user wants to buy
$7$ units of $X$ at a desired maximum price of $20$ --- the case for sell is symmetric.
\begin{enumerate} 
\item 
Consider the following two outcomes: 
1) the user gains $5$ units of $X$ at an average price of $10$;
and 2) the user gains $6$ units of $X$ at an average price of 
$11$. 
In both cases, the order is not completely fulfilled.  
The second outcome is better for the user since it can
be viewed as first achieving the first outcome  
and then paying $16$ for an extra unit, which is less
than the ask $20$.
\item 
Now, consider the 
following two outcomes: 
1) the $7$ units are completely fulfilled at an average price of $10$;
and 2) the user obtains $8$ units of $X$ at an average price of $12$.
The second outcome can occur if the user is strategic
and does not honestly report its desired amount or price. 
The first outcome
is better for the user, since the second outcome  
can be viewed as  
achieving the first outcome 
and then purchasing an extra unit  
at a marginal price of $12 \cdot 8 - 7 \cdot 10 = 26$, which is higher
than the ask $20$. 
\item 
Finally, consider the following two outcomes 
1) the $7$ units are completely fulfilled at an average price of $10$;
and 2) the user obtains $8$ units of $X$ at an average price of $11$.
The second outcome can be viewed as 
achieving the first outcome and additionally 
purchasing an extra unit  
at a price of $11 \cdot 8 - 7 \cdot 10 = 18$. 
In this case, the two outcomes
are incomparable unless
we can quantify the utility  
the user gains from the extra unit 
that it did not plan for initially.  
\end{enumerate}

\elaine{TODO: change the name weak and strong, since strong does not imply
weak}

We now formally define this natural partial ordering.

\myparagraph{A natural partial ordering.}
We use a tuple $T = (t, \amt, r, \aux)$
to represent the type of a user, where $(t, \amt, r)$
denotes the user's true demand and valuation, and $\aux$
denotes any auxiliary information.
For example, suppose $T = ({\sf Buy}(X), \amt, r, \_)$, it means
the user wants to obtain $\amt$ units of $X$ at an exchange
rate of $r$ or better.

To define the natural partial ordering, we will focus on 
the case of ${\sf Buy}(X)$/${\sf Sell}(X)$ orders; and the case for 
${\sf Buy}(Y)$/${\sf Sell}(Y)$ 
orders are symmetric, except that we replace
the usage of the rate $r$ with $1/r$.
\ke{it does not just replace $r$, but also switch the prime asset and the other?}
Based on a user's type $T = (t, \amt, r, \_)$, 
we can represent the user's true demand as
$(\delta x, r)$ where 
\[\delta x = 
\begin{cases}
 v & \text{ if $t = {\sf Buy}(X)$} \\
 -v & \text{ if $t = {\sf Sell}(X)$} 
\end{cases}
\]

We can use a pair $(\delta x_0, \delta y_0)$ to denote the outcome, 
meaning that the user has a net gain of $\delta x_0$ 
in $X$, and it has a net gain of $\delta y_0$ 
in $Y$. If $\delta x_0$ or $\delta y_0$ is negative, it means that the user
has a net loss 
in $X$ or $Y$.
\ke{I added here}
Next, we give the rules for the partial ordering $\preceq_T$.
When the context is clear, we may omit the subscript $T$.
Given two outcomes 
$(\delta x_0, \delta y_0)$ and $(\delta x_1, \delta y_1)$
and the true demand 
$(\delta x, r)$, 
we define the following partial ordering 
between
$(\delta x_0, \delta x_0)$ and $(\delta x_1, \delta y_1)$:
\begin{enumerate} 
\item 
If $\delta x_0 \geq \delta x_1$ and $\delta y_0 \geq \delta y_1$, then  
$(\delta x_0, \delta y_0) \succeq_T (\delta x_1, \delta y_1)$.
\item 
If 
$\delta x_0 \cdot  (\delta x_0 - \delta x) \leq 0$, 
$\delta x_1 \cdot  (\delta x_1 - \delta x) \leq 0$, 
and  
$r \delta x_0 + \delta x_1 \geq 0$, then we say that 
$(\delta x_0, \delta y_0) \succeq_T (\delta x_1, \delta y_1)$.

\item 
If $(\delta x_0 -  \delta x) (\delta x_1 -  \delta x) \geq 0$, 
$|\delta x_0 - \delta x| \leq |\delta x_1 - \delta x|$, 
and 
$\delta y_1  - \delta y_0 \leq r \cdot (\delta x_0 - \delta x_1)$, 
then $(\delta x_0, \delta y_0) \succeq_T (\delta x_1, \delta y_1)$.
\item 
Finally, the transitivity rule holds, that is, 
if there exists an intermediate outcome 
$(\delta x', \delta y')$ such that 
$(\delta x_1, \delta x_1) \succeq_T (\delta x', \delta y')$ and 
$(\delta x', \delta y') \succeq_T
(\delta x_0, \delta y_0)$, then it holds that 
$(\delta x_1, \delta y_1) \succeq_T (\delta x_0, \delta y_0)$. 
\end{enumerate}

In particular, the first rule 
roughly says a user always prefers 
an outcome in which it gains 
at least as much in either asset.
The second rule can be interpreted as follows:
the pre-conditions
$\delta x_0 \cdot  (\delta x_0 - \delta x) \leq 0$
and
$\delta x_1 \cdot  (\delta x_1 - \delta x) \leq 0$
mean that we have satisfied part of the demand 
(i.e., between $0$ and the entirety of the demand). 
The rule says that if this is the case, 
then the partial ordering is decided
by the natural utility function $r\delta x_0 + \delta y_0$.
The third rule 
can be interpreted as follows:
the part $(\delta x_0 -  \delta x) (\delta x_1 -  \delta x) \geq 0$
means that 
$\delta x_0$ and $\delta x_1$ are on the same
side of $\delta x$, the part
$|\delta x_0 - \delta x| \leq |\delta x_1 - \delta x|$
means that $\delta x_0$ is closer to the goal $\delta x$ than
$\delta x_1$.
The entire third rule says that 
if 
relative to $(\delta x_1, \delta y_1)$, 
the outcome
$(\delta x_0, \delta y_0)$
makes progress towards satisfying the goal 
while enjoying a {\it marginal} price of $r$ or better (where 
better means sell high or buy low), then 
$(\delta x_0, \delta y_0) \succeq_T (\delta x_1, \delta y_1)$.
Conversely, if relative to 
$(\delta x_0, \delta y_0)$, 
the outcome
$(\delta x_1, \delta y_1)$ 
goes in the opposite direction
of satisfying the goal, 
while suffering from a {\it marginal} price of $r$ or worse, then  
$
(\delta x_1, \delta y_1) \preceq_T 
(\delta x_0, \delta y_0) 
$.
The last rule is the standard transitivity rule for 
any partial ordering relation.

We can mechanically verify that all of the above rules are internally consistent.

\elaine{TODO: i changed the partial ordering, change everything else that depends on it.}

\section{No Short-Selling Variant}

In this section, we consider a variant of our scheme 
with the following modifications:
\begin{enumerate} 
\item 
The mechanism checks the balance 
of a user (in either asset) and ensures that the user's account balance
does not go negative. 
In other words, no user is allowed to short-sell;
\item Suppose the user's belief
of the rate $r'$ is greater than the market rate $r$, 
then the user's honest strategy would be to sell as much $Y$ 
as possible 
as long as the market rate stays below $r'$; and vice versa.
\end{enumerate}

For this no-short-selling variant, it is natural
to define a total 
ordering among outcomes. 
Specifically, for a user whose belief of the rate
is $r'$, given a position of $(x, y)$, i.e., when it holds $x$ 
units of $X$ and $y$ units of $Y$,
then the value of  
the position can be calculated as
$r' x + y$.
It is easy to verify that this total ordering
is a refinement of the partial ordering
defined in \Cref{sec:rank}.
We shall prove that incentive compatibility 
holds with respect to this natural total ordering.
Specifically, since no outcomes are incomparable
under a total ordering, incentive compatibility 
(\Cref{defn:ic}) now simply means
that {\it the honest strategy 
maximizes a user's utility}.

\subsection{Definition}
We first make a slight modification to the syntax 
of the swap mechanism. 

\myparagraph{Partial fulfillment swap mechanism without short-selling.}
Recall that every user has a position 
${\sf Pos}(x, y)$ that denotes
its balance in $X$ and $Y$, and we require that $x \geq 0, y \geq 0$.

Partial fulfillment swap mechanism without short-selling
is defined similarly as the partial fulfillment mechanism 
of \Cref{sec:defn}, 
except with the following modifications.
We now additionally assume that the mechanism's allocation rule
has an extra input which is 
the current positions of all users. 
Further, each user has a unique identifier 
that is included in the $\aux$ field of the order, such that the mechanism
can see who submitted the order. 
Note that there can be multiple orders coming from the same user.
The mechanism is required to produce an outcome 
such that no user's ending position is negative
in either $X$ or $Y$.

We require 
incentive compatibility to hold
no matter what the users' initial positions are. 

\begin{figure*}[t]
\begin{mdframed}
    \begin{center}
    {\bf Our swap mechanism: variant with no short-selling} 
    \end{center}

\vspace{2pt}
\noindent
    \textbf{Input:} 
A current pool state ${\sf Pool}(x_0, y_0)$, 
and a vector of orders 
${\bf b}$, and all users' initial positions. 

\vspace{5pt}
\noindent\textbf{Mechanism:}
\begin{enumerate} 
\item 
Let $r_0 := r(x_0, y_0)$
be the initial rate.
Ignore all ${\sf Buy}(X)$/${\sf Sell}(Y)$  
orders whose specified rate 
$r < r_0$, and ignore all 
${\sf Buy}(Y)$/${\sf Sell}(X)$ orders 
whose specified rate $r > r_0$.
Let 
${\bf b}'$ 
be the remaining orders.
\item 
Attempt to safe-execute all orders in ${\bf b'}$ 
\elaine{TODO: define safe-execute}
at rate $r_0$ (without actually executing them): 
if the sum of the net gain in $X$ of all users is non-negative,  
we call it the ${\sf Buy}(X)/{\sf Sell}(Y)$ case;
otherwise
we call it the ${\sf Buy}(Y)/{\sf Sell}(X)$ case.
\ignore{
Let $\sigma = \sum_{(t, \amt, r) \in {\bf b}'} \beta(t, \amt, r)$
where $\beta(t, \amt, r) = \begin{cases}
\amt & \text{ if } t = {\sf Buy}(X)\\
-\amt & \text{ if } t = {\sf Sell}(X)\\
-\amt/r_0 & \text{ if } t = {\sf Buy}(Y)\\
\amt/r_0 & \text{ if } t = {\sf Sell}(Y)
\end{cases}$

We call $\sigma \geq 0$ the ${\sf Buy}(X)$/${\sf Sell}(Y)$-dominant case,
and $\sigma < 0$ the  
${\sf Buy}(Y)$/${\sf Sell}(X)$-dominant case.
}
\item 
The ${\sf Buy}(X)$/${\sf Sell}(Y)$-dominant case: 
\begin{enumerate} 
\item 
\label{stp:sort-noshort}
Sort ${\bf b}'$ such that all the ${\sf Sell}(X)$
and ${\sf Buy}(Y)$
orders 
appear in front of the 
${\sf Sell}(Y)$
and ${\sf Buy}(X)$ orders. 
Write the resulting list 
of orders as $\{(t_i, \amt_i, r_i, ({\id}_i, \_))\}_{i \in [n']}$.

\item 
Attempt to safe-execute 
orders in $\{(t_i, \amt_i, r_i, ({\id}_i, \_))\}_{i \in [n']}$
sequentially (without actually 
executing them), 
let $\beta_i$ denote the net gain in $X$
contributed by the $i$-th order.
\item 
Henceforth we assume that there exists an index 
$j \in [n']$ such that 
$\sum_{i = 1}^j \beta_i = 0$.
If not, 
we can find 
the smallest index 
$j \in [n']$ 
such that $\sum_{i = 1}^j \beta_i > 0$, 
and split the $j$-th order into two orders 
$(t_j, \amt_{j, 0}, r_j, ({\id}_i, \_))$
and $(t_j, \amt_{j, 1} , r_j, ({\id}_i, \_))$, 
resulting in a new list with $n'+1$
orders, such that  $\amt_{j, 0} + \amt_{j, 1} = \amt_j$, 
and moreover, 
index $j$ of the new list satisfies 
this condition.

\item 
Phase 1: Safe-execute the first $j$ orders at rate $r_0$.
\item 
Phase 2: Safe-execute each remaining order $i \geq j + 1$ in sequence.
\end{enumerate}

\item 
The ${\sf Buy}(Y)$/${\sf Sell}(X)$-dominant case is symmetric.  
\ignore{
and we describe it 
below for completeness.
\begin{enumerate} 
\item 
Sort ${\bf b}'$ such that all the 
${\sf Sell}(X)$
and ${\sf Buy}(Y)$
orders 
appear {\color{blue}after} the 
${\sf Sell}(Y)$
and ${\sf Buy}(X)$ orders. 
Write the resulting list 
of orders as $\{(t_i, \amt_i, r_i)\}_{i \in [n']}$.
\item 
Find the smallest index $j \in [n']$ 
such that ${\color{blue}-}
\sum_{i = 1}^j 
\beta(t_i, \amt_i, r_i) \geq 0$.
\item 
Fully execute the first $j-1$ orders 
and the first $\amt_j 
{\color{blue}+} \sum_{i = 1}^j \beta(t_i, \amt_i, r_i)$
units of the $j$-th order   
at an exchange rate of $r_0$.
\item 
Replace the $j$-th order 
$(t_j, \amt_j, r_j)$
with the unfulfilled portion, that is, 
$(t_j, \amt_j, r_j) \leftarrow 
(t_j, {\color{blue}- } \sum_{i = 1}^j \beta(t_i, \amt_i, r_i), r_j)$.
\item 
For each $i \geq j + 1$ in sequence, 
fulfill as much of the $i$-th remaining order as possible,
that is, 
pick the largest $\amt \leq \amt_i$ such that  
subject to the constant-function market maker $\potential$, 
the new market exchange rate \elaine{TODO: define}
$r {\color{blue} \geq } r_i$ 
if $\amt$ units are to be executed;  
execute $\amt$ units of the $i$-th order.
\end{enumerate}
}
\end{enumerate}
\end{mdframed}
\caption{Our swap mechanism: variant without short-selling}
\label{fig:mech-noshort}
\end{figure*}

\subsection{A Swap Mechanism Without Short-Selling}
We now describe a variant of our earlier mechanism but now providing
an additional guarantee of no short-selling.  
Basically, the new mechanism is almost the same as the old one,
except that during the execution, 
the mechanism always ensures that 
after (partially) executing every  
order, the corresponding user's  
position
is non-negative in both $X$ and $Y$.

Henceforth, we can use the notation 
$(t, \amt, r, ({\id}, \aux))$
to denote an order where $t$, $\amt$, and $r$
are the same as before, 
${\id}$ denotes the user's pseudonym, and 
$\aux$ denotes any arbitrary auxiliary information.
To precisely capture no short-selling, we define 
the notion of 
{\it safe-execute} below:

\elaine{TODO: associate line numbers with phases in narrative}

\begin{enumerate} 
\item {\it Safe-execute} in Phase 1: 
Do the following based on the type of the order:
\begin{itemize} 
\item $({\sf Buy}(X), \amt, r, (\id, \_))$: 
Pick the maximum possible $\amt' \leq \amt$ such that 
buying $\amt'$ units of $X$ will not cause
user $u$'s balance in $Y$ to go negative;
and let the user $u$ buy $\amt'$ units of $X$
at a fixed rate $r_0$. 
\item $({\sf Sell}(X), \amt, r, (\id, \_))$: 
Pick the maximum $\amt' \leq \amt$ such that 
selling $\amt'$ units of $X$ will not cause
user $u$'s balance in $X$ to go negative;
and let the user $u$ sell $\amt'$ units of $X$ at a fixed rate $r_0$. 
\item 
$({\sf Buy}(Y), \amt, r, (\id, \_))$ or 
$({\sf Sell}(Y), \amt, r, (\id, \_))$: 
symmetric except that we now use $1/r$ in place of $r$.
\end{itemize}
\item {\it Safe-execute} in Phase 2: 
Do the following based on the type of the order, 
where $r_{\rm cur}$ denotes the market rate before
the order is executed:
\begin{itemize} 
\item $({\sf Buy}(X), \amt, r, (\id, \_))$: 
Pick the maximum $\amt' \leq \amt$ such that 
if the user $\id$ bought
$\amt'$ amount of $X$ at a starting rate of $r_{\rm cur}$, 
then 1)  the ending rate $r_{\rm end} \leq r$; and 2) 
user $u$'s balance in $Y$ will not go negative.
Let the user $\id$ buy $\amt'$ amount of 
$X$ at the starting market rate $r_{\rm cur}$.
\item $({\sf Sell}(X), \amt, r, (\id, \_))$: 
Pick the maximum $\amt' \leq \amt$ such that 
if the user $\id$ sold 
$\amt'$ amount of $X$ at a starting rate of $r_{\rm cur}$, 
then 1) the ending rate $r_{\rm end} \geq r$; and 2) 
user $u$'s balance in $X$ will not go negative.
\item 
$({\sf Buy}(Y), \amt, r, (\id, \_))$ or 
$({\sf Sell}(Y), \amt, r, (\id, \_))$: 
symmetric except that we now use $1/r$ in place of $r$.
\end{itemize}
\end{enumerate}

\begin{theorem}[Arbitrage resilience]
\label{thm:strong-arbitrage-util-model}
The swap mechanism in \Cref{fig:mech-noshort} satisfies arbitrage resilience.
In particular, this holds no matter how ties are broken
during the sorting step.
\end{theorem}
\begin{proof}
The proof is the same as that of \Cref{thm:strong-arbitrage}.
We explain it for the 
${\sf Buy}(X)$/${\sf Sell}(Y)$-dominant case, 
since the 
${\sf Buy}(Y)$/${\sf Sell}(X)$-dominant case
is symmetric.
The mechanism executes in the following way: 
In phase 1, 
it (partially) executes a set of orders all at the 
initial rate $r_0$, such that there is no change to the initial pool state 
${\sf Pool}(x_0, y_0)$. In phase 2, it executes only 
${\sf Buy}(X)/{\sf Sell}(Y)$ orders, which enjoy rate that is $r_0$ or greater, due to increasing marginal cost (\Cref{fact:inc-marginal-cost}).
Therefore, for any subset of orders,
it cannot be the case that 
there is a net gain in one asset without any loss in the other.
\end{proof}

\ignore{
\begin{theorem}[Fair treatment]
If the sorting step
in \Cref{fig:mech-noshort}
breaks ties
by random ordering, then 
the resulting mechanism satisfies fair treatment.
\end{theorem}
\begin{proof}
By the same proof as of \Cref{thm:fairness}.
\end{proof}
}

\myparagraph{Refined swap mechanism with no short-selling in the weak fair-sequencing model.}
In the weak fair-sequencing model, we consider
the mechanism described in \Cref{fig:mech}
with the following refinements:
\begin{itemize} 
\item 
The sorting step in Line~\ref{stp:sort-noshort} 
will break ties 
using arrival order. 
\item 
A user's honest strategy is defined as follows: 
$HS(t, \amt, r, \aux)$
simply outputs 
a single order 
$(t, \amt, r, \aux)$, where $t$ is either ${\sf Sell}(X)$ or ${\sf Sell}(Y)$, and $\amt$ is the amount of $X$ or $Y$ held by the user, respectively.
\end{itemize}

\begin{theorem}
Suppose $\potential$ is concave, increasing, and differentiable.
In the weak fair-sequencing model, 
the above refined swap mechanism in \Cref{fig:mech-noshort} 
satisfies incentive compatibility 
w.r.t. the total ordering.
\label{thm:strongic}
\end{theorem}

\subsection{Proof of Incentive Compatibility}
Observe that for the mechanism in \Cref{fig:mech-noshort}, the facts of increasing marginal cost (\cref{fact:inc-marginal-cost}), no free lunch (\cref{fact:no_free_lunch}), and \cref{fact:prime_asset} \ke{give it a name? prime asset switch?} still hold.
Since the total ordering is a refinement of the partial ordering, 
by the same proof as in \cref{lem:coalesce} and \cref{lem:earlygood}, we have the following result:
\begin{lemma}
\label{lem:coalesce-noshort}
Suppose $\potential$ is concave, increasing, and differentiable.
Given any initial state ${\sf Pool}(x_0,y_0)$ and order vector ${\bf b}_{-u}$,
for any strategic order vector ${\bf b}_u$ of user $u$, there exists
a single order $b'_u$ 
such that 
1) $b'_u$ results in an outcome at least as good as 
${\bf b}_u$ w.r.t. the total ordering; 
2) the arrival time used in $b'_u$  
is the true arrival time $\alpha^*$ of user $u$;
and 3) 
either $b'_u = (\_, 0, \_, \_)$ 
or 
$b'_u$ would be completely 
safe-executed
under ${\sf Pool}(x_0, y_0)$ and ${\bf b}_{-u}$.
\end{lemma}
\begin{proof}
    By the same proof of \cref{lem:coalesce} and \cref{lem:earlygood}. 
    Suppose the joint outcome of ${\bf b}_u$ is $(\delta x, \delta y_0)$.
    Throughout the proof,
    the strategic order vector ${\bf b}_u$ is replaced with a single order $b'_u$ that can be fully secure-executed and results in an outcome $(\delta x, \delta y_1)$.
    Since $\delta y_1\geq \delta y_0$ by the same proof, 
    the outcome of $b'_u$ is at least as good as ${\bf b}_u$ w.r.t. the total ordering.
\end{proof}
Thus,
it suffices to consider strategies that submit a single order,
declare the true arrival time $\aux^*$, 
and moreover, either the order has a $0$ amount or
it will be completely executed under ${\sf Pool}(x_0, y_0)$ and ${\bf b}_{-u}$.
Henceforth, we call such strategies 
as {\it admissible, single-order} strategies.
The proof of \Cref{thm:strongic} can be completed by showing the following lemma.

\begin{lemma}
\label{lemma:strong_single_strategic}
Suppose $\potential$ is concave, increasing, and differentiable.
For any admissible and single-order strategy~$S$, 
the honest strategy results in an outcome that is at least
as good as strategy~$S$.  
\end{lemma}

\begin{proof}
The structure of the proof is similar to~\Cref{lemma:single_strategic}.
However, since the original partial order
is refined to a total order,
some pair of incomparable outcomes in the previous proof
will now become comparable.
We prove it for the case when user $u$'s type is  
$({\sf Sell}(X), \amt^*, r^*, \aux^*)$.
The case for ${\sf Sell}(Y)$ is symmetric.

Due to \Cref{lem:coalesce-noshort} and \cref{fact:prime_asset}, we may assume
that the strategic order must be 
of the type either 
${\sf Buy}(X)$ or ${\sf Sell}(X)$, with a true time of arrival $\aux^*$.
Henceforth, 
let $(\delta x, \delta y)$ and $(\delta x', \delta y')$ denote the honest and strategic outcomes, respectively.

\myparagraph{Case 1:} 
Suppose the strategic order is of the type ${\sf Buy}(X)$.
Observe that the honest order
will always generate an outcome $(\delta x, \delta y)$
such that $r^* \cdot \delta x + \delta y \geq 0$.
Hence, it suffices to argue that
the outcome $(\delta x', \delta y')$
generated by any ${\sf Buy}(X)$ 
must have $r^* \cdot \delta x' + \delta y' \leq 0$.
Note that $\delta x' \geq 0$, we have $\delta y' \leq 0$
because of no free lunch fact (\Cref{fact:no_free_lunch}).

 Since the user's true type is ${\sf Sell}(X)$,
it means the initial market rate $r_0$ satisfies $r^* \leq r_0$.
Observe that in our mechanism, a ${\sf Buy}(X)$-type order can only be executed at a rate $r_0$ or larger.
This means that $-\delta y' \geq r_0\cdot \delta x'$.
Thus, it must be that  $r^* \cdot \delta x' + \delta y' \leq 0$ since $r^* \leq r_0$.

Therefore, we can conclude that the honest outcome is
at least as good as the strategic outcome.

\myparagraph{Case 2:} 
The strategic order 
is in the same direction of ${\sf Sell}(X)$.
Because of~\Cref{fact:prime_asset},
we may assume that the strategic order is also of type ${\sf Sell}(X)$.
Since the mechanism guarantees no short-selling, for any strategic ${\sf Sell}(X)$ order, it must be that
$\delta x'\geq -\amt^*$ and $\delta x \geq -\amt^*$.
We consider the following cases. 
\begin{itemize} 
\item 
{\it Case 2a:} $0 \geq \delta x' \geq \delta x$. 
Since the 
strategic order declares the same arrival time as the honest one by its admissibility,
if the orders from both strategies get safe-executed
for a non-zero amount, both executions will start at the same market exchange rate. 

Therefore, for selling the initial $|\delta x'|$ units of
$X$, the two strategies are equivalent.
The honest strategy sells an additional 
$|\delta x - \delta x'|$ units of $X$ at rates of at least $r^*$,
i.e., $\delta y - \delta y' \geq r^* \cdot (\delta x' - \delta x)$.
This means that
$r^* \cdot \delta x + \delta y =
(r^* \cdot \delta x' + \delta y') + r^* \cdot (\delta x - \delta x')
+ (\delta y - \delta y') \geq r^* \cdot \delta x' + \delta y'$.
This implies that the honest outcome is
at least as good as the strategic outcome.

\item  {\it Case 2b:}
$0 \geq \delta x > \delta x'$. 
Because of no short selling,
the strategic outcome must also satisfy $\delta x' \geq -\amt^*$.

In this case, the honest outcome has not reached the goal $-\amt^*$.
Under the honest strategy, 
after user $u$'s order has been safe-executed, 
the state of the market is such that
if a further non-zero portion of the order is executed,
this portion will incur an average rate of strictly less than $r^*$.

Hence, for the strategic order, 
the difference of $|\delta x' - \delta x| > 0$
units must be traded at an average rate strictly worse less
than $r^*$,
i.e., $\delta y' - \delta y < r^* \cdot (\delta x - \delta x')$.
In other words,
$r^* \cdot (\delta x' - \delta x)
+ (\delta y' - \delta y) < 0$,
i.e., the strategic outcome is strictly worse than the honest outcome.

\end{itemize}

\end{proof}

\end{document}